\documentclass[]{article}

\title{Lightweight Robust Framework for\\Workload Scheduling in Clouds}

\author{Muhammed Abdulazeez, Dariusz R. Kowalski, Prudence W.H. Wong\\Department of Computer Science, University of Liverpool, UK, \\Email: [m.abdulazeez,d.kowalski,pwong]@liverpool.ac.uk \\Pawel	Garncarek\\Institute of Computer Science, Wroclaw University, Poland,
	\\Email:pgarn@cs.uni.wroc.pl}

\usepackage{amsmath}
\usepackage{amsthm}
\usepackage{graphicx}
\usepackage{subfigure}
\usepackage{mathrsfs}
\usepackage{algorithm}
\usepackage{color}
\usepackage{url}
\usepackage{enumerate}
\usepackage{booktabs}
\usepackage{amsmath}
\usepackage{amssymb}
\usepackage[english]{babel}
\usepackage{algpseudocode}
\usepackage{epsfig}
\usepackage{float}
\usepackage{makeidx}
\usepackage[normalem]{ulem}

\newtheorem{theorem}{Theorem}
\newtheorem{lemma}{Lemma}

\newcommand{\ep}{\varepsilon}

\newcommand{\remove}[1]{}

\DeclareMathOperator*{\argmax}{arg\,max}

\newcommand{\pg}[1]{{#1}}
\newcommand{\ba}[1]{{#1}}

\newcommand{\comment}[1]{}

\newcommand{\SecureMaxWork}{RobustMaxWork}

\begin{document}

\maketitle

\begin{abstract}
Reliability, security and stability of cloud services without sacrificing too much resources have become a desired feature in the area of workload management in clouds. 
The paper proposes and evaluates a lightweight framework for scheduling a workload 
which part could be unreliable. This unreliability could be caused by various types
of failures or attacks.
Our framework for robust workload scheduling efficiently combines classic 
fault-tolerant and security tools, such as packet/job scanning, with 
workload scheduling, and it does not use any heavy resource-consuming 
tools, e.g., cryptography or non-linear optimization. More specifically, the framework uses a novel objective function to allocate
jobs to servers and constantly decides which job to scan based on
a formula associated with the objective function.
We show how to set up the objective function and the corresponding scanning procedure 
to make the system provably stable, provided it satisfies a specific stability condition. 
As a result, we show that our framework assures cloud stability
even if naive scanning-all and scanning-none strategies are not stable.
We extend the framework to decentralized scheduling and evaluate it under several
popular routing procedures.
\end{abstract}

\section{Introduction}
\label{sec:intro}


Cloud computing~\cite{MellG2011} enables ubiquitous, convenient, on-demand network access to a shared pool of configurable computing resources.
It is becoming more and more popular for businesses to access computing facilities without investing in IT infrastructure. ~\cite{ec2,icloud,computeEngine,azure}.
Cloud users send in resource requests in an online manner and the cloud provider allocates the required resources for the required amount of time. What is important, the resource allocation is transparent to the users.
The provider allocates resources based on the system dynamicity and current system load.
We refer readers to surveys on cloud computing for cloud technology~\cite{FosterZRL2008, ArmbrustFG+2009, ZhangCB2010}.

While there is a growth in the use of cloud services, many potential users are still reluctant to deploy their business in the cloud.
Major concerns are its reliability, security and stability~\cite{SubashiniK2011}.
There are different reliability and security issues depending on different delivery models of cloud services, including Software as Service (SaS), Platform as Service (PaS) and Infrastructure as Service (IaS).
In this work we focus on the IaS model. 
This technology makes the users and provider reside at different locations and virtually access the resources over the Internet, therefore any security concerns threatening the Internet also threaten the cloud.
In particular, we consider the scenarios when part of the workload is unreliable,
e.g., fault-prone or generated by malicious sources, and propose a lightweight framework that combines load management and detection of unreliable traffic.
We investigate how to strike a balance between efficient workload scheduling and packet/job scanning so that we can maintain stability (as a guarantee of bounded buffers at machines) without sacrificing too much resources to filter out the unreliable part of the workload. IaS provides users with  computing infrastructure in the form of Virtual Machines (VM). Following~\cite{MaguluriSY12}, we assume that the users request resources such as memory, CPU and storage, for a certain amount of time in the form of VMs; this corresponds to a job to be done.
Upon receiving the requests (typically in a form of packets), the system has to allocate the required resources by scheduling the VMs on the server.
We extend the model in~\cite{MaguluriSY12} by considering scenarios where part of the 
workload is genuine and the other unreliable. Genuine traffic comes from real users; completing these requests counts towards system's work done. Unreliable traffic is subject to failures or comes from attackers, who aim to disrupt the system by issuing requests that occupy resources; completing these does not count as proper work done. We adopt a classic reliability and security tool of packet scanning to detect these malicious packets~\cite{ModiPB+13}. While scanning is able to distinguish genuine from unreliable requests,it consumes and wastes resources that would normally be used for serving genuine workload. 

On the other hand, as we do not know whether the packets are faulty/fake until we scan them, 
we may also waste time and resources in scanning genuine packets. Therefore, the scheduling algorithm needs to strike a balance between the resources wasted by scanning and by performing unreliable requests without scanning them.

We consider centralized and distributed scheduling algorithms. 
In the {\em centralized setting}, there are central queues, and upon arrival  
jobs are added to the central queues corresponding to the requested type of VMs;
recall that there is a limited number of types of VMs (as each of VMs is
in fact a small operating system~\cite{MaguluriSY12}) and each job is allocated
to a VM of the requested type. When the resources become available, the centralized scheduling algorithm determines which set of jobs is to be served and to which servers the VMs are mapped to.
In the {\em distributed setting}, each server has its own queues; upon arrival, a job request is forwarded  to some server and stored in the server's local queue corresponding to the requested
type of VMs. The requests from these local queues are served by a distributed scheduling algorithm locally on the server.

The system is \emph{stable} if the queues do not tend to increase without bound.
We aim to characterize the maximum arrival rates of genuine and unreliable requests under which there is an algorithm to maintain the stability of the system and to develop such algorithm if it exists.
In addition, to guarantee quality of service, 
we also measure job latency, 
which is defined as the amount of time a job resides in the system since its arrival.
We present the precise model in Section~\ref{sec:model} and the proposed algorithms and analysis in Sections~\ref{sec:algorithm} to~\ref{sec:decentralized}. 
We present the experiments and conclude in Sections~\ref{sec:experiments}
and~\ref{sec:miscellaneous}.

\subsection{Related Work}
\label{sec:relatedwork}

Apart from maintaining stability,
there are many other design issues related to 
workload management in cloud computing.
Cloud utilization has been considered in~\cite{XieJYZ2013,RezvaniAJ2015,YazirMF+2010}. 
Optimizing other costs of running the services has been considered~\cite{StolyarZ2015,WangMZ2011,XieJYZ2013}. 

The algorithms we propose here
are inspired by the MaxWeight algorithm 
analyzed in~\cite{tassiulas1992stability} in the context of scheduling genuine workload only, and could be seen as its efficient generalizations to unreliable environments.
The MaxWeight algorithm has been since investigated extensively~\cite{MarkakisMT2013,SunZW+2013,MaguluriTY2014}. 
Detecting and distinguishing unreliable or malicious from genuine requests and 
a number of  approaches have been proposed~\cite{ModiPB+13,BakshiD10}.
In this paper, we assume that such a tool to scan a packet and detect potentially unreliable or malicious packages is available.
The authors in~\cite{MaguluriS14} studied jobs with unknown duration
and analyzed several decentralized approaches and showed that some are throughput-optimal while others are not. 
Another study~\cite{hung2014task} aimed to optimize
recovery time after failures; it is different from our aim to prevent the impact of failures by 
tailored combination of scheduling and scanning tools.

\subsection{Our Contributions}
\label{sec:contribution}

We propose a lightweight robust framework to manage workload in clouds
under unreliable workload scenarios.
Extending the model in~\cite{MaguluriSY12,tassiulas1992stability},
we propose to detect unreliable part of the traffic by scanning only some 
specifically selected jobs without sacrificing too much resources. 
\begin{itemize}
\item
We propose a theoretical model to capture the essence of this conditional scanning
and show that under a certain system capacity region and stochastic arrival pattern of genuine and unreliable jobs there exists an algorithm, called {\SecureMaxWork}, that manages the workload while maintaining queue stability (i.e., the queue is bounded and does not grow to infinity size).
\item
We show how to efficiently compute the optimal scanning strategy (vector).
\item
We prove that there is no stable algorithm for the workload outside the capacity region
for which {\SecureMaxWork} is stable.
\item
We propose several distributed versions of {\SecureMaxWork} and discuss various extensions of theoretical results.
\item 
We evaluate the algorithms and the proposed model using extensive simulations,
with respect to the maximum and average latency over time.
\end{itemize}

\comment{
\begin{itemize}
\item We represent unreliable traffic as a variable fraction of ``genuine'' traffic and consider how varying the rate of unreliable traffic affects the system performance.
\item We consider a simple scanning mechanism and characterize the system capacity region 
such that it is possible to realize traffic loads under this region and the queues of the systems can be stabilized. The capacity region in unreliable scenario is defined in much complex way than that  
in the idealistic reliable model~\cite{MaguluriSY12}. 
\item We propose an algorithm, called {\SecureMaxWork}, that guarantees stability within this capacity region; which
 
	is a balance of MaxWeight-type of opitimization 
	(on the total pending workload instead of the number of pending jobs) with
	carefully selected scanning strategy.
	We also adapt {\SecureMaxWork} to the distributed setting. 
	The stability guarantees of the algorithms
	are analyzed mathematically.
\item We then evaluate the algorithm and the proposed model using extensive simulations.
  We evaluate both the maximum and average queue sizes and latency over time.

\end{itemize}
}

\section{Model}
\label{sec:model}

We consider a cloud system modeled by a network of physical machines that have limited available resources (for instance, CPU, memory, storage, \dots) and is supposed to be able to process an ongoing stream of jobs.

\textbf{Servers.} We consider a set of $n$ networked servers (physical machines). Each server has its own resources that it can distribute among jobs, that is, for each
resource it has a fixed capacity.

\textbf{Jobs.} 
A job is specified by its type and length. \pg{Since there are limited number of virtual machine types, we only consider limited number of job types -- there are $J$ types of jobs.} Each type is a set of demands for resources; more specifically, for each available server resource a job type has a number specifying how much of this resource is required in order to process any job of this type.

There are $I$ different lengths of jobs possible: $L_1,\ldots,L_I$

We consider online random arrival model, where new jobs arrive independently of each other and are identically distributed across all time slots,
and the variance of arrival length is finite. Let $\lambda_{i,j}$ denote expected \pg{sum of lengths} of genuine (i.e., user-generated) type-$j$ jobs of length $L_i$ that arrive per time slot, for any positive integers $j\le J$ and $i\le I$.

\textbf{Processing jobs and feasible configurations.}
Each server can process a set of jobs simultaneously, as long as the cumulative amount of each resource used by these jobs does not exceed the server capacity for this resource. 
Processing jobs is done in synchronous time steps, also called rounds.
The whole system capacity is a linear sum of capacities of all the servers.
Given job types and server capacities, one can compute the set $\mathcal{S}$ of all feasible configurations, where feasible configuration denotes a vector 
$N=(N_1,\ldots,N_J)$ such that the system can process 
simultaneously $N_j$ type-$j$ jobs, for every $j$.

\textbf{Malicious jobs and security tools.} 
Let $\kappa_{i,j}$ denote the expected \pg{sum of lengths} of malicious jobs of type-$j$ of 
length $L_i$ that arrive per time slot. 
Similarly as genuine jobs, malicious jobs arrive independently of each other and are identically distributed across all time slots, and the variance is finite.
We assume that we have a scanning tool that, given a job, can detect whether it is a genuine user request (we will call it a  {\em good job}) or a malicious request (a {\em malicious job}). Each scanning takes $1$ time slot per job and requires same resources as the original job \pg{(scanning is done on the same virtual machine)}.

\textbf{Central scheduler.} We consider a central scheduler with a queue of all injected, but not yet finished, jobs. The scheduler decides which servers process which jobs for the next time slot. After this time slot, all unfinished jobs return to the scheduler with saved progress and can be processed further at a later time and by a different server. This property of a system is called preemptiveness.
The centralized algorithm, called {\SecureMaxWork}, will be introduced in Section~\ref{sec:algorithm}.

\textbf{Distributed scheduler.} In decentralized approach all servers maintain separate queues for jobs of type $j$, therefore when a job arrives decision has to be made as to which server to route the job to. Each server runs locally a protocol {\SecureMaxWork} with respect to its local queues. The distributed implementations of algorithm {\SecureMaxWork} will be presented in Section~\ref{sec:decentralized}.

\textbf{Notation.}
 
\begin{itemize}
\item
$n$ denotes the number of servers in the cloud; 
\item
$I$ denotes the number of different job lengths;
\item
$J$ denotes the number of different job types;
\item 
$A(t)=(A_1(t),\ldots,A_J(t))$ denotes the vector of sets of type-$j$ jobs, for $j\le J$, which arrive to the system in the beginning of time slot $t$;
\item
$Q(t)$ denotes the vector of queue lengths (i.e., sum of lengths of jobs in the queue) for each type of jobs in the beginning of time slot $t$;
\item 
$Q_j(t)$ denotes the total length of users' and malicious type-$j$ jobs, for $j\le J$, in the beginning of time slot $t$;
\end{itemize}

In addition there are the following notations regarding {\SecureMaxWork} algorithm:
\begin{itemize}
\item
$\alpha_{i,j}$ is a probability of scanning type-$j$ job of length $L_i$; 
the algorithm may implement a specific scanning strategy, i.e., 

use a specific vector $\alpha$.
\item
$X_j(t)$ is the total length of queued type-$j$ jobs that will not be scanned,
taken in the beginning of time slot $t$ 
(i.e., the algorithm scanned them already or decided not to scan them at all); 
\item
$Y_j(t)$ is the total length of queued type-$j$ jobs that will be scanned, 
taken in the beginning of time slot $t$
(i.e., the algorithm has already decided to scan them, but has not scanned them yet); 
\item
$Z_j(t) = Z_j(Q(t))$ is the expected time required to process type-$j$ jobs stored in queue in the beginning of time slot $t$; \pg{the formula with an explanation for it will be given in section \ref{sec:algorithm};}

\item $a_j$ is the expected time required to process type-$j$ jobs that arrive in one time slot; \pg{the formula with an explanation for it will be given in section \ref{sec:alg-feasbility};}

\end{itemize}
Whenever time slot $t$ is clearly fixed or understood from the context, we  may omit an argument $t$ from the formulas.

\textbf{Scanning strategies.} We will compare the following scanning strategies:
\begin{itemize}
\item 
Scan-None--- always executes a job without scanning, i.e., $\alpha_{i,j}=0$ for all $i,j$;
\item Scan-All --- scans all jobs except those with processing time shorter or equal to the scanning time (recall that scanning takes $1$ time slot), i.e., $\alpha_{i,j}=0$ for $L_i\leq 1$ and $\alpha_{i,j}=1$ otherwise;
\item 
Scan-Opt--- will be defined in section \ref{subsec:opt_scan_freq}.
\end{itemize}

\textbf{Stability.} We say that, given arrival rates $\lambda$ and $\kappa$, the algorithm is stable if the expected queue size at any fixed time is bounded, i.e. $\limsup \limits_{t\rightarrow \infty} E[\sum_j Q_j(t)] < \infty$.

\section{Main Algorithm --- Centralized Version}
\label{sec:algorithm}

\begin{algorithm}
\ba{Recall that $\lambda$ is the expected length of genuine jobs,
$\kappa$ is the expected length of malicious jobs and $\alpha$ is probability of scanning jobs.}
\begin{algorithmic}
\caption{{\SecureMaxWork}($\lambda, \kappa, \alpha$)}

\State $X \gets \vec{0}$ {\tt // \ba{jobs that will not be scanned}}
\State $Y \gets \vec{0}$ {\tt // \ba{jobs be that will be scanned}}
\State $Q \gets \vec{0}$ {\tt // \ba{all jobs}}
\Loop
	\State new time slot begins
	\ForAll {new type-$j$ job $\tau_{i,j}$ of length $L_i$}
		\State $r \gets$ random value from $[0;1]$
		\If {$r <\alpha_{i,j}$}  \hfill {\tt // $\tau_{i,j}$ to be scanned}
			\State $Y_j \gets Y_j + L_i$ 
			\State $Q_j \gets Q_j + L_i$
		\Else   \hfill {\tt // $\tau_{i,j}$ not to be scanned}
			\State $X_j \gets X_j + L_i$ 
			\State $Q_j \gets Q_j + L_i$
		\EndIf
	\EndFor
	\ForAll {$j$} 
		\State $Z_j \gets X_j + Y_j(\lambda_j/(\lambda_j+\kappa_j) + E(1/L_j))$
	\EndFor
	\State $N' \gets \argmax_{N \in \mathcal{S}} \sum_j N_j \cdot Z_j$
	\ForAll{$j$}
		\For{$k \leq N'_j$}
			\State Process\_job($j$)
		\EndFor
	\EndFor
\EndLoop

\label{alg:SecureMaxWork}

\end{algorithmic}
\end{algorithm}

\begin{algorithm}
\begin{algorithmic}
\caption{Process\_job($j$)}

\If {there are still jobs in $X_j$ and $Y_j$ that are not yet scheduled to be processed in this time slot}
	\State $r \gets$ random value from $[0;1]$
	\If {$r < X_j/(X_j+Y_j)$}
		\State Process\_job\_X($j$)
	\Else
		\State Process\_job\_Y($j$)
	\EndIf
\ElsIf {there are no jobs in $X_j$ that are not yet processed in this time slot, but there are still such jobs in $Y_j$}				
	\State Process\_job\_Y($j$)
\ElsIf {there are no jobs in $Y_j$ that are not yet processed in this time slot, but there are still such jobs in $X_j$}				
	\State Process\_job\_X($j$)
\Else
	\State Stay Idle
\EndIf

\label{alg:Procedure_job}

\end{algorithmic}
\end{algorithm}

\begin{algorithm}
\begin{algorithmic}
\caption{Process\_job\_X($j$) {\tt \ba{// not to be scanned}}}

\State Process a unit of any unscheduled job contributing to $X_j$
\State $X_j \gets X_j -1$
\State $Q_j \gets Q_j -1$

\label{alg:Procedure_X}

\end{algorithmic}
\end{algorithm}

\begin{algorithm}
\begin{algorithmic}
\caption{Process\_job\_Y($j$){\tt \ba{// to be scanned}}}

\State Scan any unscheduled job contributing to $Y_j$
\State $L_i \gets$ length of the scheduled job
\State $Y_j \gets Y_j - L_i$
\If {detected as malicious}					
	\State $Q_j \gets Q_j - L_i$
\Else
	\State $X_j \gets X_j + L_i$
\EndIf

\label{alg:Procedure_Y}

\end{algorithmic}
\end{algorithm}

Algorithm {\SecureMaxWork} (see Algorithm \ref{alg:SecureMaxWork}, with Algorithms~\ref{alg:Procedure_job}, \ref{alg:Procedure_X} and \ref{alg:Procedure_Y} as sub-procedures) is parametrized by: 
scanning vector $\alpha=(\alpha_{i,j})_{i\le I, j\le J}\in [0,1]^{I\times J}$, vector of rates of genuine user's requests $\lambda=(\lambda_{i,j})_{i\le I, j\le J}$, and vector of rates of malicious requests $\kappa=(\kappa_{i,j})_{i\le I, j\le J}$ .
Upon arrival of type-$j$ job of length $L_i$, algorithm {\SecureMaxWork} decides to scan it with probability $\alpha_{i,j}$ (c.f., the first {\bf for all} loop in 
Algorithm~\ref{alg:SecureMaxWork}).

\pg{The key idea of {\SecureMaxWork} is to measure the expected time required to process all jobs of each type $j$ and prioritize the jobs of type which accumulated the most. The expected time (also called expected work) required to process all jobs of type $j$ accumulated in queue at time $t$ is denoted by $Z_j(t)$.}

\pg{It takes $X_j$ time to process jobs that will not be scanned (jobs contributing to $X_j$). Jobs contributing to $Y_j$ will need to be scanned (by definition of $Y_j$), which requires $Y_j \cdot E(1/l_j)$ expected time. In expectance $\lambda_j/(\lambda_j+\kappa_j)$ fraction of scanned jobs are genuine, so after scanning, they still must be processed, taking in total $Y_j \cdot \lambda_j/(\lambda_j+\kappa_j)$ time. $\kappa_j/(\lambda_j+\kappa_j)$ fraction of scanned jobs are fake and after scanning they take no more processing time. Therefore: \\
$Z_j(t) = X_j(t) + Y_j(t) \cdot (\lambda_j/(\lambda_j+\kappa_j) + E(1/\ell_j))$,\\
where $\ell_j$ is a (random) length of arriving type-$j$ jobs.}

The algorithm then computes values $Z_j$ (c.f., the second {\bf for all} loop in 
Algorithm~\ref{alg:SecureMaxWork}) and
finds configuration $N$ from the set of feasible server configurations $\mathcal{S}$ 
that maximizes the sum $\sum_{j=0}^J Z_j(t) N_j$, i.e., the objective of the algorithm in each
time slot $t$ is:
$$\max_{N \in S} \sum_{j=0}^J Z_j(t) N_j \ .$$
This configuration is denoted by $N'$. \pg{The quick intuition behind this function is that the more jobs of a given type accumulate, the more weight should be put to scheduling the jobs of that type in order to prevent further accumulation. $Z_j$ here is the weight given to jobs of type $j$.}

Finally, in the last {\bf for all} loop, the algorithm processes $N'_j$ jobs of type $j$,
for each $j \in \{1,\dots,J\}$;
that is, from each processed job it executes a unit of it and the total size of
$Q_j$ decreases by $N'_j$ at the end of time slot $t$. 
It is done by calling procedure Process\_job($j$), c.f., 
Algorithm~\ref{alg:Procedure_job}.
If $N'_j$ is larger than the number of different type-$j$ jobs in the queues, {\SecureMaxWork} 
processes as many type-$j$ jobs as possible instead, each time processing
a unit of each such job.
(c.f., the second part of procedure Process\_job($j$)). If $N'_j$ is smaller than the number of different type-$j$ jobs in the queues, {\SecureMaxWork} has to decide which type-$j$ jobs to process (c.f., the first part of procedure Process\_job($j$)). 
It repeats $N'_j$ times: 

\begin{itemize}
\item with probability $X_j/(X_j+Y_j)$ it processes a job that will not be scanned 
(i.e., a job that contributes to $X_j$),
\item with probability $Y_j/(X_j+Y_j)$ it scans a job pending for scanning
(i.e., a job that contributes to $Y_j$).
\end{itemize}
If there are not enough jobs contributing to $X_j$, it processes all jobs contributing to $X_j$ and as many jobs contributing to $Y_j$ as possible, so that altogether it processes $N_j$ type-$j$ jobs. Vice versa, if there are not enough jobs contributing to $Y_j$, it processes all jobs contributing to $Y_j$ and as many jobs contributing to $X_j$ as possible.
Processing and/or scanning a specific type-$j$ job is done by calling sub-procedures
Process\_job\_X($j$) and/or Process\_job\_Y($j$), respectively (c.f., 
Algorithms~\ref{alg:Procedure_X} and \ref{alg:Procedure_Y}, resp.)
directly from the execution of procedure Process\_job($j$).

In short words, we could describe a single time slot of an execution of algorithm 
{\SecureMaxWork} as follows. 
We always have $X_j + Y_j = Q_j$, for any type-$j$, as each job is either
waiting for scanning or not (i.e., has been already scanned or is not selected for
scanning at all). 
Whenever a type-$j$ job of length $L_i$ arrives, with probability $\alpha_{i,j}$ its length is added to $Y_j$, otherwise its length is added to $X_j$. When the algorithm executes one unit of a job contributing to $X_j$, $X_j$ is reduced by $1$. When the algorithm executes one unit of job contributing to $Y_j$, it means it scans it --- if it was a genuine user job, its length is removed from $Y_j$ and added to $X_j$ \pg{(so the algorithm spent one round on scanning, but the sum $X_j + Y_j$ remains the same)}; if it was a malicious job, its length is removed from $Y_j$.

\section{Analysis}
\label{sec:analysis}

In this section we prove that algorithm {\SecureMaxWork}
is stable if there is $\ep>0$ and vector $a$ such that 
$a \in (1-\ep) \cdot co(S)$.

\begin{theorem}
\label{stability}
The {\SecureMaxWork} algorithm is stable for all arrival patterns $\lambda, \kappa$, for which there exists a stable algorithm.
\end{theorem}

In the remainder of this section we will prove Theorem~\ref{stability}. 
We will need the following result (extension of Foster's criteria for irreducible Markov chains).

\begin{theorem}[\cite{ASM1987}]
\label{th:Foster}
Consider a Markov chain $Q(t)$ with state space $\mathcal{Q}$. Consider a random walk on it, starting from a state $x$. Let $\tau_x$ denote the time when the random walk first reaches some recurrent state (or infinity if it never reaches any). If there exists a lower bounded real function $V: \mathcal{Q} \rightarrow \mathbb{R}$, an $\epsilon > 0$ and a finite subset $\mathcal{Q}_0$ of $\mathcal{Q}$ such that
\begin{equation}
E[V(Q(t+1)) - V(Q(t)) | Q(t) = q] < -\epsilon, \quad \text{if } q \notin \mathcal{Q}_0,
\end{equation}
\begin{equation}
E[V(Q(t+1)) | Q(t) = q] < \infty, \quad \text{if } q \in \mathcal{Q}_0, 
\end{equation}
then we have
\begin{equation}
P(\tau_q < \infty) = 1, \quad 
\forall q \in T
\end{equation}
and all states $q \in \cup_{j=1}^\infty R_j$ are positive recurrent.
\end{theorem}

Let $V(Q(t)) = \sum_j (Z_j(Q(t)))^2$. Note that $V(Q(t)) \geq 0$ for all possible queue states $Q(t) \in \mathcal{Q}$. We show that there exist two positive numbers $b,\epsilon$ such that the inequality
\begin{equation}
E[V(Q(t+1)) - V(Q(t)) | Q(t) = q] < -\epsilon
\end{equation}
holds for all $q \in \mathcal{Q}$ for which $q_j \geq b$.

Let $A(t)$ denote the vector of arrival lengths for each type of job, with distinction between jobs that will be scanned and jobs that will not be scanned, in the beginning of time slot $t$. Let $Alg(t)$ denote the vector of queue changes due to algorithm decisions for each type of job, with distinction between jobs that will be scanned and jobs that will not be scanned, in the beginning of time slot $t$.
We will be using $Z_j$ as a shorthand of $Z_j(Q(t))$, $A$ as a shorthand of $A(t+1)$, and $Alg$ as a shorthand for $Alg(t+1)$. 

$$\begin{array}{rl}
& \hspace*{-3em} E[V(Q(t+1)) - V(Q(t)) | Q(t) = q] \\
=& E[\sum_j[Z_j(Q(t+1))^2 - Z_j^2] | Q(t)=q] \\
=& E[\sum_j[(Z_j + Z_j(A) - Z_j(Alg))^2 - Z_j^2] | Q(t)=q] \\
=& E[\sum_j[(Z_j(A) - Z_j(Alg))^2 +\\
& + 2Z_j(A - Z_j(Alg)) ] | Q(t)=q] \\
\leq & K + 2E[\sum_j[Z_j \cdot Z_j(A)] |Q(t)=q] + \\
& - 2E[\sum_j[[Z_j \cdot Z_j(Alg) ] | Q(t)=q] \ .
\end{array}$$
The last inequality comes from $E[\sum_j[(Z_j(A) - Z_j(Alg))^2] | Q(t)=q]$ being upper bounded under assumption that the variances of arrival lengths are finite;
 
we denoted this upper bound by $K$. 

\begin{lemma}
\label{lem:finite}
There exists finite set $\mathcal{F} \subseteq \mathcal{Q}$ such that for all $q \in \mathcal{Q}-\mathcal{F}$:
$$\begin{array}{rl}
K + 2E[\sum_j[Z_j \cdot Z_j(A)] |Q(t)=q] + & \\
& \hspace*{-9em} - 2E[\sum_j[[Z_j \cdot Z_j(Alg) ] | Q(t)=q] < 0 \ .
\end{array}
$$
\end{lemma}

To prove Lemma~\ref{lem:finite}, we need the following result:

\begin{lemma}
\label{lem:aux}
For almost all queue states $q$ there exists a feasible configuration $N \in S$ such that $N \cdot Z \geq K + a \cdot Z$ (where $Z = Z(q)$ and $\cdot$ is scalar product).
\end{lemma}

\begin{proof}
We remind that vector $\alpha$ is chosen in such a way that the vector $(1+\epsilon)a$ lies inside convex hull of set $\mathcal{S}$\pg{, further denoted by $co(S)$} (set $\mathcal{S}$ is the set of feasible server configurations $N$), where
$$a_j = (\lambda_j + \kappa_j)((1-\alpha_j) + \alpha_j(\dfrac{\lambda_j}{\lambda_j + \kappa_j} + E(\dfrac{1}{L}))) \ .$$

Let $a'$ be a vector (a point) corresponding to intersection of vector $a$ with a face $F$ of $co(S)$. $a' \geq (1+\epsilon)a$ is a linear combination of some feasible configurations $N^{(1)},\dots,N^{(k)}$ (configurations on the face $F$). Therefore for any non-negative vector $Z$ there exists a configuration $N$ such that $N \cdot Z \geq a' \cdot Z$ (at least one of $N^{(1)},\dots,N^{(k)}$ is such a configuration). So $N \cdot Z \geq (1+\epsilon)a \cdot Z$, thus $N \cdot Z - \frac{\epsilon}{1+\epsilon}N \cdot Z \geq a \cdot Z$.

Consider the set of queue states $\mathcal{Q}'=\{q: \exists N \ \frac{\epsilon}{1+\epsilon}N \cdot Z(q) \geq K\}$.
\begin{equation}
\forall q \in \mathcal{Q}' \exists N \in S \ N \cdot Z(q) \geq K + a \cdot Z(q) \ .
\end{equation}

In order to complete the proof of the lemma, it remains to show that $\mathcal{Q}-\mathcal{Q}'$ is finite. 

Note that $Z_j$ is a monotonically increasing function and if direction of $q$ is same as $q'$ then direction of $Z(q)$ is same as $Z(q')$. For each possible direction of $Z(q)$, for all $q \in \mathcal{Q}$ different from $\overrightarrow{0}$, there exists $q' \in \mathcal{Q}'$ such that $Z(q')$ is in the same direction as $Z(q)$ ($q'$ may be some multiplicity of $q$ that is large enough to ``kill'' constant $K$). Then for any $q'' \geq q'$ we have $Z(q'') \geq K + b \cdot Z(q'')$.

For each possible direction of $Z$ we can take the minimum vector $q_0 \in \mathcal{Q}'$ such that $Z(q_0)$ is in the considered direction. Then we take one vector $q_0'$ that is greater than all $q_0$'s for each direction. 

In each direction there is a finite number of vectors smaller than the corresponding 
$q_0$. 
The lengths of vectors $q_0$ are bounded, as the optimized function is continuous 
and considered on a compact set of directions.

Hence the number of configurations of smaller length than the supremum $q_0$
is finite, and so $\mathcal{Q}-\mathcal{Q}'$ is finite. 

\end{proof}

\begin{proof}[Proof of Lemma~\ref{lem:finite}]
Note that in Lemma~\ref{lem:finite}, $E[\sum_j[Z_j \cdot Z_j(A)]]$ is the expected value of the scalar product $Z \cdot Z(A)$ and $E[\sum_j[[Z_j \cdot Z_j(Alg) ]$ is the expected value of the scalar product $Z \cdot Z(Alg)$. According to Lemma~\ref{lem:aux}, inequality from lemma~\ref{lem:finite} is true for almost all queue states, i.e., there exists a finite set of states $\mathcal{F}$ such that for all $q \in \mathcal{Q}-\mathcal{F}$ the desired inequality holds.
\end{proof}

\begin{proof}[Proof of Theorem~\ref{stability}]
According to Lemma~\ref{lem:finite} and Theorem~\ref{th:Foster}, given arrival rates for which there exists some stable algorithm, {\SecureMaxWork} algorithm reaches positive recurrent state in a finite time, therefore it is stable.
\end{proof}

\section{Determining feasible capacity region and scanning frequency}
\label{sec:alg-feasbility}

\subsection{Feasible capacity region}

\begin{lemma}
\label{lemma:processing}
Processing a type-$j$ job for 1 time slot decreases $Z_j(Q(t))$ by 1 on average.
\end{lemma}
\begin{proof}
If the processed step was not scanning (i.e., processing a job from $X_j$) then trivially $Z_j$ decreased by 1.

If the processed step was scanning a job of length $L_i$ (i.e., processing a job from $Y_j$), then:
\begin{itemize}
\item before scanning that job contributed $\frac{\lambda_j}{\lambda_j+\kappa_j}\cdot L_i + 1$ weight towards $Z_j$;
\item with probability $\frac{\lambda_j}{\lambda_j+\kappa_j}$ it was a genuine job, so after scanning it contributes $L_i$ towards $Z_j$ (increase in weight);
\item with probability $\frac{\kappa_j}{\lambda_j+\kappa_j}$ it was a malicious job, so after scanning it contributes $0$ towards $Z_j$ (decrease in weight).
\end{itemize} 
Therefore, on average $Z_j$ decreases by 
$$\begin{array}{rl}
& \hspace*{-7em}
\frac{\lambda_j}{\lambda_j+\kappa_j}(1-\frac{\kappa_j}{\lambda_j+\kappa_j}L_i) + \frac{\kappa_j}{\lambda_j+\kappa_j}(1+\frac{\lambda_j}{\lambda_j+\kappa_j}L_i) = \\
= & 1 + \frac{\lambda_j \kappa_j}{(\lambda_j+\kappa_j)^2}(-L_i + L_i) = 1 \ .
\end{array}$$%
\end{proof}

Recall that $\alpha_{i,j}$ is the probability of scanning type-$j$ job of length $L_i$, 
and $A(t)$ denotes the vector of arrival lengths for each type of job, with distinction between jobs that will be scanned and jobs that will not be scanned in the beginning of time slot $t$. 

Let $a_j = a_j(\alpha, \lambda, \kappa) = E[Z_j(A(t))]$ be the expected weight of 
type-$j$ jobs that arrive per time slot 
(arrivals are i.i.d. across time slots, so $E[A_j(t)]=E[A_j(t+1)]$ for all $t$). 
Then 
$$a_j = \sum_i p_{i,j}[(\lambda_j + \kappa_j)(1-\alpha_{i,j}) + \alpha_j(\lambda_j + \kappa_j)(\dfrac{\lambda_j}{\lambda_j + \kappa_j} + \dfrac{1}{L_i})],$$ 
where $p_{i,j}$ is the probability that type-$j$ job has length $L_i$. Addend ($\lambda_j + \kappa_j)(1-\alpha_{i,j})$ corresponds to the weight of good and malicious jobs that will not be scanned (so it contributes to $X_j$). Addend $\alpha_j (\lambda_j + \kappa_j) \cdot \frac{\lambda_j}{\lambda_j + \kappa_j}$ corresponds to the weight of good jobs that will be scanned but without scanning taken into account yet (so it contributes to $Y_j$). Addend $\alpha_j (\lambda_j + \kappa_j)\cdot 1/L_i$ corresponds to the weight of scanning good and malicious jobs (so it also contributes to $Y_j$).
Let $a = (a_1, \dots, a_J)$.

\begin{theorem}
\label{non-stability}
If arrivals $\lambda$ and $\kappa$ are such that for all vectors $\alpha$ arrivals $a \notin co(\mathcal{S})$ then no algorithm is stable.
\end{theorem}
\begin{proof}
Consider arrival rates $\lambda$, $\kappa$ and scanning probabilities $\alpha$ such that $a \notin co(\mathcal{S})$. 

We will show that for every algorithm there exists $j$ such that $E[Z_j(Q(t))]$ is unbounded.

In every time slot $t$ the weight of queues $Z(Q(t))$ is changing on average by
$E[Z(Q(t+1)) - Z(Q(t))] = E[Z(Q(t)+A(t+1)-Alg(t+1)) - Z(Q(t))] = E[Z(A(t+1))-Z(Alg(t+1))] = a - E[Z(Alg(t+1))] $.

Note that $a \notin co(S)$, while $E[Z(Alg(t+1)] = N(t+1) \in co(S)$ (according to Lemma \ref{lemma:processing}).
If we consider multiple time slots and any combination of algorithm decisions $N(t)$ then the vector of weights of queues $Z(Q(t))$ is growing in the direction of vector 
$a$. Therefore there exists $j$ such that $Z_j(Q(t))$ is unbounded with regard to $t$. Therefore $Q_j(t)$ is unbounded with regard to $t$, which is contradictory with our definition of stability.
\end{proof}


\subsection{Optimal scanning frequencies}
\label{subsec:opt_scan_freq}

\begin{theorem}

If there exists a vector of scanning frequencies $\alpha^{(0)}$ such that given job arrival rates $\lambda$ and $\kappa$ are inside the capacity region (i.e., $(1+\epsilon)a(\alpha^{(0)},\lambda, \kappa) \in co(S)$ as defined in section \ref{sec:alg-feasbility}), then there exists a vector of scanning frequencies $\alpha^{(1)} \in \{0,1\}^{I \times J}$ such that these job arrivals are inside the capacity region ($(1+\epsilon)a(\alpha^{(1)},\lambda, \kappa) \in co(S)$).
\end{theorem}
\begin{proof}
Recall the following properties that we will be using in the proof.
First, scanning uses same resources as the scanned job.
Second, if vector $N \in S$, then also $N' \in S$ for all $0 \leq N' \leq N$. Therefore, if $a \in co(S)$ then $a' \in co(S)$ for all $0 \leq a' \leq a$. (Inequalities between vectors $a \leq b$ mean that for all $i$, $a_i \leq b_i$.)

Assume vector $a(\alpha^{(0)},\lambda, \kappa) \in co(S)$, where
$a_j(\alpha,\lambda, \kappa) = \sum_i p_i[(\lambda_j + \kappa_j)(1-\alpha_{i,j})
+ \alpha_{i,j}(\lambda_j + \kappa_j)(\dfrac{\lambda_j}{\lambda_j + \kappa_j} + 1/L_i)))]$,
as defined in section \ref{sec:alg-feasbility}. Therefore, function $a_j(\alpha,\lambda, \kappa)$ is independent of $\lambda_k, \kappa_k, \alpha_{i,k}$ for $k \neq j$ and for all $i$. Given fixed $\lambda$ and $\kappa$, $a_j(\alpha,\lambda, \kappa)$ is a linear combination of scanning frequencies $\alpha_{i,j} \in [0,1]$, for all $i$. 

Therefore the vector $\alpha_j = (\alpha_{1,j},\alpha_{2,j},\dots,\alpha_{I,j})$ that minimizes $a_j$ is one of the extreme points of region $[0,1]^I$. Furthermore, we can easily compute $\alpha_{i,j}^{(1)}$ for all $i$ that minimize $a_j$, since each summand $p_i[(\lambda_j + \kappa_j)(1-\alpha_{i,j}) + \alpha_{i,j}(\lambda_j + \kappa_j)(\dfrac{\lambda_j}{\lambda_j + \kappa_j} + 1/L_i)))]$ is independent of all other summands. Value $\alpha_{i,j}^{(1)}$ that minimizes this summand is $0$, if $1 \leq \dfrac{\lambda_j}{\lambda_j + \kappa_j} + 1/L_i$, and $1$ if $1 > \dfrac{\lambda_j}{\lambda_j + \kappa_j} + 1/L_i$, for each $i,j$, and can be computed using Algorithm~\ref{alg:frequencies}.

This gives $a(\alpha^{(0)},\lambda, \kappa) \geq a(\alpha^{(1)},\lambda, \kappa))$, which means that $(1+\epsilon)a(\alpha^{(1)},\lambda, \kappa)) \in co(S)$.
\end{proof}

\begin{algorithm}
\begin{algorithmic}
\caption{Optimal\_scanning\_frequencies($\lambda, \kappa$)}

\ForAll{$i,j$}
	\If {$\lambda_{i,j}/(\lambda_{i,j}+\kappa_{i,j}) + 1/L_i \geq 1$}
		\State $\alpha_{i,j}\gets 0$
	\Else
		\State $\alpha_{i,j}\gets 1$
	\EndIf
\EndFor

\label{alg:frequencies}
\end{algorithmic}
\end{algorithm}

\section{Decentralization}
\label{sec:decentralized}

Centralized scheduler uses the same queue for all type-$j$ jobs that are waiting in the system. In decentralized approach each server maintains its own queues for jobs of type-$j$, therefore when a job arrives a decision has to be made as to which server to route the job. In this section we specify and analyze six different decentralized implementations of the main algorithm {\SecureMaxWork} from Section~\ref{sec:algorithm}with different routing procedures.

Another distinct property of decentralized approach is using the main {\SecureMaxWork} algorithm at each server to make scheduling decision. 

In~\cite{MaguluriS14} the authors introduced the notion of \emph{refresh time}.  
A time slot $t$ is a  \emph{global refresh time} if there is no job currently in queues in all the servers at the beginning of $t$. \emph{Local refresh time} occurs when there is no job currently in service at the beginning of time slot $t$ in a server $v$. 
Note that a global refresh time happens when all servers have their local refresh times.
In practice global refresh times occurs rarely and intuitively it happens to be more rear as the number of servers increase. Due to this phenomenon, all our
decentralized algorithms will rely only on the local refresh times at servers.

{\em Decentralized {\SecureMaxWork} for each server $v$.} Upon local refresh time, server $v$ keeps idle until some of its queues
are non-empty. Then it applies {\SecureMaxWork} with respect of its queues
to find configuration $N'$. Then it keeps scheduling jobs according to this configuration
every round until the next local refresh time (i.e., there is no suitable job to 
apply configuration $N'$. This scheme is repeated.

Below we describe six specification of the decentralized {\SecureMaxWork}
with different routing policies. \ba{Algorithm 1 and 3 were analysed in ~\cite{MaguluriS14}, 4 was analysed in ~\cite{MaguluriSY12}, and we designed 2, 5, and 6 in this work. We will discuss the performance of the new algorithms and their comparisons with the existing ones in the results section.}

{\em Algorithm 1: {\SecureMaxWork}\_JSQ.} Joint Shortest Queue (JSQ) paradigm is used to route a newly arrived job, that is,
it is sent to the server with the queue with the smallest number of jobs of type-$j$,
where $j$ denotes the type of the arrived job. 
This algorithm was analyzed in the context of cloud workload in~\cite{MaguluriS14}.

{\em Algorithm 2: {\SecureMaxWork}\_JSW.} Joint Shortest Work (JSW) is used for routing a newly arrived job; i.e., 
it is sent to the server with the minimum work load of type-$j$, where 
the work load is defined as a sum of lengths of jobs stored in the local queue
of type-$j$.

{\em Algorithm 3: {\SecureMaxWork}\_UR.}  Uniformly Random (UR) routing is used for forwarding newly arrived job: each job that arrives into the system is routed to one of the servers chosen uniformly at random. 
This algorithm was analyzed in the context of cloud workload in ~\cite{MaguluriS14}

{\em Algorithm 4: {\SecureMaxWork}\_RR.}  Round Robin (RR) routine is used for allocating newly arrived jobs to the servers: 
for each type-$j$ there is a pointer going cyclically along the servers, showing
which was the last server to which a type-$j$ job was allocated; when a new type-$j$ job arrives it is sent to the next server (modulo $n$) and the pointer is advanced to that server.

{\em Algorithm 5: {\SecureMaxWork}\_P2Q.}  
Power of two Choices combined with selection of the Shortest Queue (P2Q)
is used for routing a newly arrived job of a type-$j$: two servers are sampled uniformly at random, and the job is routed to the server with the shorter type-$j$ queue.

{\em Algorithm 6: {\SecureMaxWork}\_P2W.} Power of two Choices combined with selection of the Shortest Work load (P2W) is used for routing a newly arrived job of a type-$j$: two servers are sampled uniformly at random, and the job is routed to the server with the smaller workload of type-$j$(i.e., where the total length of type-$j$ jobs in the local queue is shorter).


\section{Simulations}

\label{sec:experiments}

\subsection{Experiment Setting}
\label{sec:exp-setting}

The setup for simulations, described in this section, is based on the one in Maguluri et al.~\cite{MaguluriSY12}.

\paragraph{Servers and VMs}
\ba{We consider a server with 30 GB memory, 30 EC2 computing units and 4000G storage space.} There are $100$ identical servers in the cloud. Arriving jobs are served in the cloud based on three types of virtual machines described in Table~\ref{t:cloud}. 
\ba{This gives three maximal configurations available at each server:}
$(2,0,0)$, $(1,0,1)$ and $(0,1,1)$.

\begin{center}
\begin{table}[!htbp]
\protect\caption{Representation of Instances in Amazon EC2}
\label{t:cloud}
\begin{center}
\begin{tabular}{|c|c|c|c|}
\hline 
Instance type & Memory (GB) & vCPU & Storage (GB)\tabularnewline
\hline 
\hline 
Standard  & 15 & 8 & 1,690\tabularnewline
\hline 
High-Memory & 17.1 & 6.5 & 420\tabularnewline
\hline 
High-CPU & 7 & 20 & 1,690\tabularnewline
\hline 
\end{tabular}
\end{center}
\end{table}
\end{center}

\paragraph{Job arrivals}
We use the generic arrival vector $\lambda^*= 0.99 \cdot (1,1/3,2/3)$ for the genuine users' workload, 
which is located at the border of the server capacity area 
(it is easy to observe that it is a normalized linear
combination of the three maximal configurations, 
additionally re-scaled by factor $0.99$). 

In each time step a job of type $j=1,2,3$ is selected with probability 
$\frac{\lambda^*_j}{130.5}$, and its length is chosen according to the length 
distribution described below with the mean length $130.5$.

Similarly as above, 
 
we define a malicious workload using a generic arrival
vector $\kappa^*=(0.7,0.01,0.01)$, and the procedure of generating a malicious traffic is analogous
as above for generating the genuine users' one. 

Note that each of the arrival rates $\lambda^*$ and $\kappa^*$ is within 
the capacity range of a server, whereas the combined workflow rate $\lambda^*+\kappa^*$
is not.

\paragraph{Job size distribution}

When a new job is generated, with the probability of $0.7$ it is an integer that is uniformly distributed in the interval $[1,50]$, with the probability of $0.15$ it is an integer  uniformly distributed in the interval $[251,300]$, and with the probability of $0.15$ it is an integer uniformly distributed in the interval $[451,500]$.

Note that there are $150$ possible job lengths, and the mean length is $130.5$,
as assumed in the definition of arrival rates.

\paragraph{Set up of simulations.}

Since there are $100$ homogenous servers, 
the overall arrival rates are: $\lambda = 100\cdot \lambda^* = (99,33,66)$
for genuine workload, and $\kappa = 100\cdot \kappa^* = (70,1,1)$
for malicious workload.
The job size distribution is as specified above, same for each job type.
We computed the following optimal scanning vector $\alpha^*$ for this setting,
more precisely, the vector minimizing expected arriving weight:
\begin{itemize}
\item
$\alpha^*_{i,1} = 0$ for $L_i \leq 2$, 
\item
$\alpha^*_{i,2} = 0$ for $L_i \leq 34$, 
\item
$\alpha^*_{i,3} = 0$ for $L_i \leq 50$,
\item
$\alpha^*_{i,j} = 1$ otherwise. 
\end{itemize}

Each execution includes $4,000,000$ time steps.
We monitor the following parameters every $200,000$ time steps:
\begin{itemize}
\item
{\em Average queue size} by the recorded time step;
\item
{\em Maximum queue size} by the recorded time step;
\item
{\em Average latency} by the recorded time step;
\item
{\em Maximum latency} by the recorded time step.
\end{itemize}
In the first part, we output the results of the above measurements 
for centralized protocols: 
\begin{description}
\item[LambdaFlow:] \ \ \ \ \ 
{\SecureMaxWork} applied for genuine flow only, (i.e., only with genuine arrival rate $\lambda$), and
no scanning is applied (i.e., $\vec{\alpha}=\mathbf{0}$);
\item[ScanOPT:] \ \ \ \ \ \
{\SecureMaxWork} applied for simultaneous genuine and malicious flows, 
with scanning defined by vector $\alpha^*$;
\item[ScanALL:] \ \ \ \ \ \
{\SecureMaxWork} applied for simultaneous genuine and malicious flows, 
with scanning all jobs of size bigger than $1$ (i.e., for every $j=1,2,3$, 
$\vec{\alpha_{1,j}}=0$ and $\vec{\alpha_{i,j}}=1$ for every $i>1$);
\item[ScanNONE:] \ \ \ \ \ \
{\SecureMaxWork} applied for simultaneous genuine and malicious flows, 
with no scanning (i.e., $\vec{\alpha}=\mathbf{0}$).
\end{description}
As by theoretical part, it is expected that the first two executions
should be stable while the last one is not.
We expect that the third execution is also not stable, which would 
justify our research quest for searching of suitable scanning vector.
In order to visualize it, we also display differences and ratios between
the second and the third executions --- the stable and the potentially
unstable one.

The second part of simulations is dedicated to decentralized algorithms.
We study how different routing protocols influence stability, when applied
to the {\SecureMaxWork} with the optimally selected scanning vector $\alpha^*$.
We compare the six decentralized implementations with the centralized one,
{\SecureMaxWork}-OPT.
They are denoted by ScanOPT\_JSQ, ScanOPT\_JSW, ScanOPT\_UR, ScanOPT\_RR,
ScanOPT\_P2Q, ScanOPT\_P2W, and ScanOPT, respectively.


\subsection{Results}
\label{sec:results}


\begin{figure*}[!h]
	\centering     
	\includegraphics[scale=0.135]{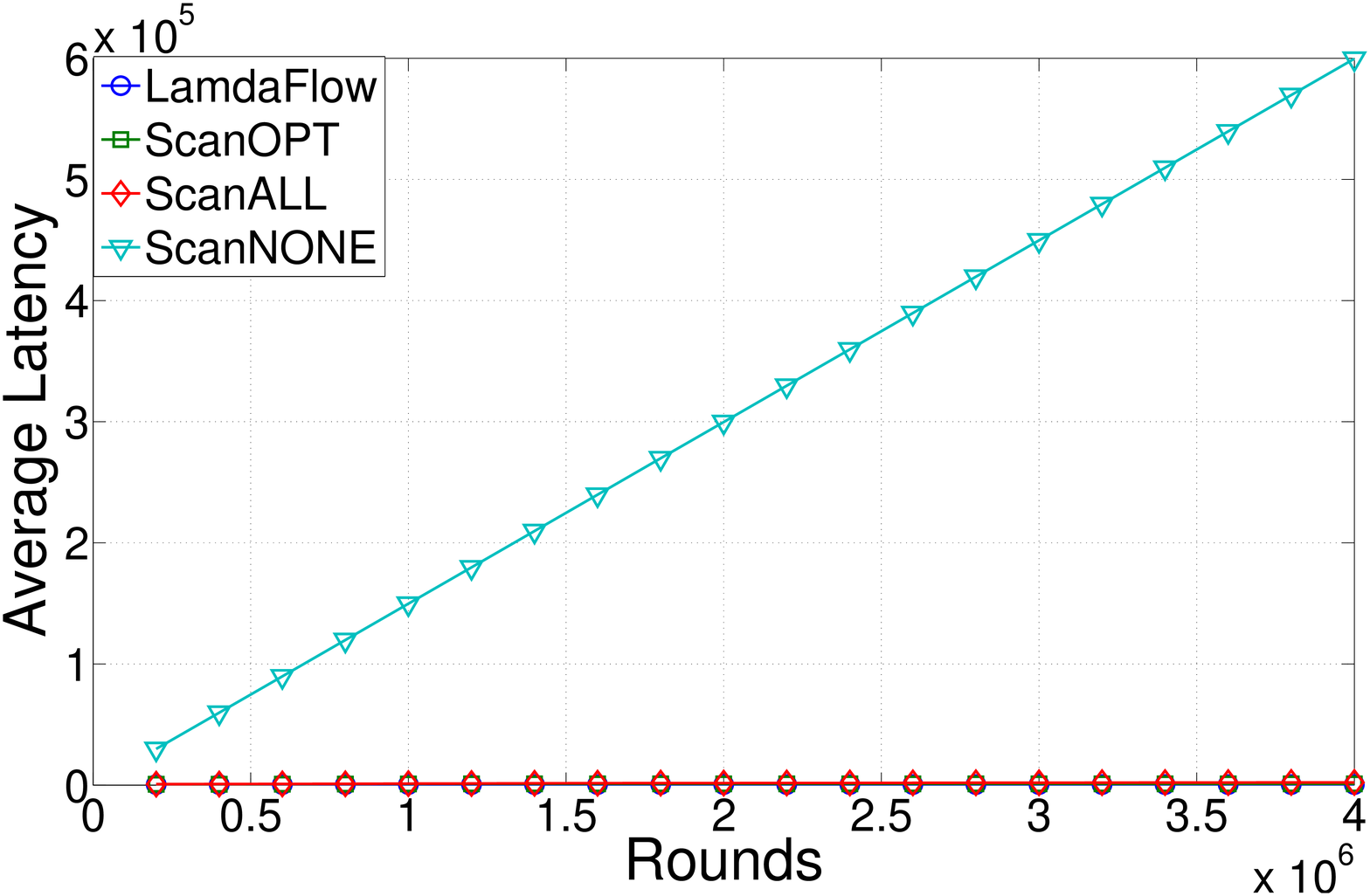}
	\includegraphics[scale=0.135]{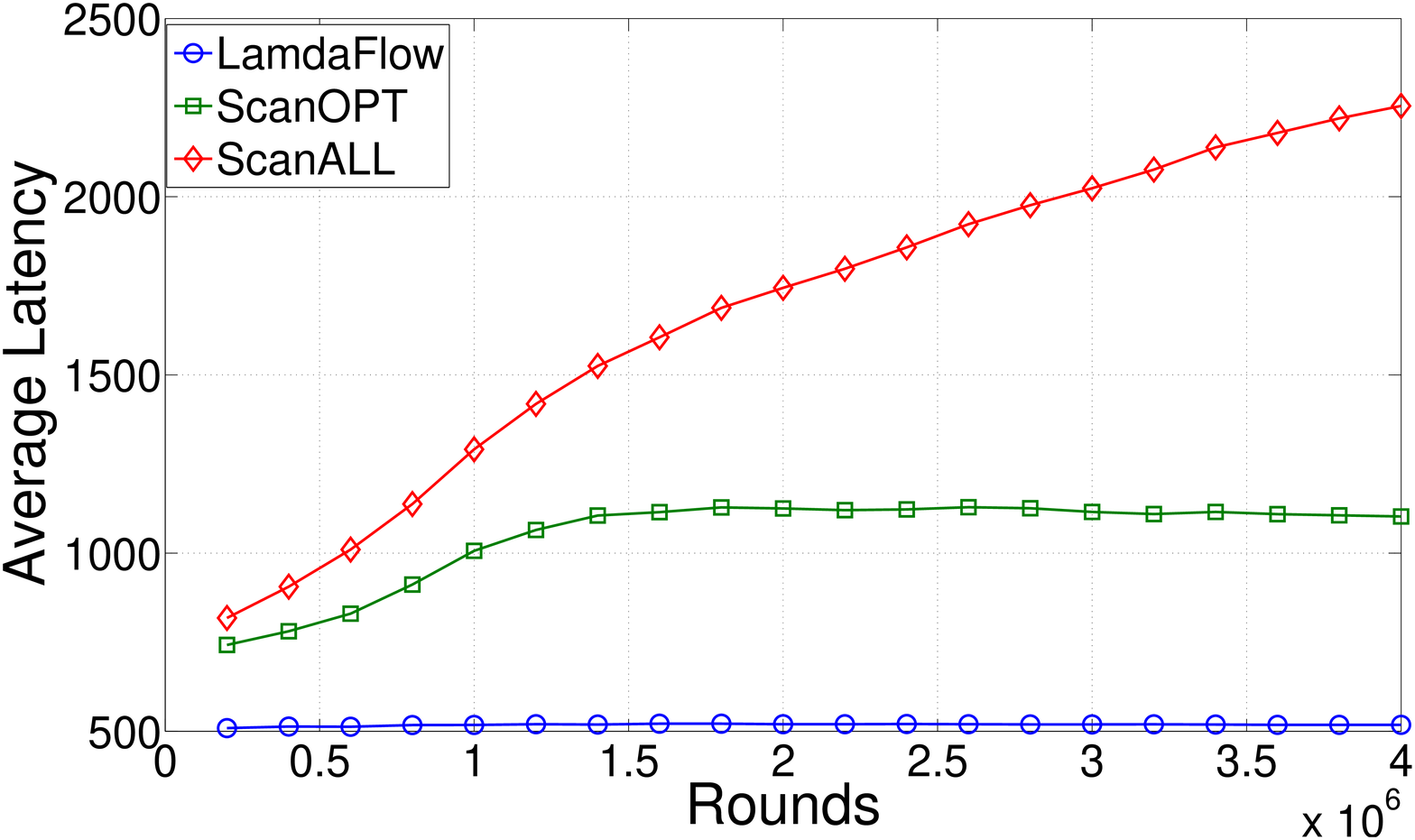}
	\caption{Comparison of average latency of LambdaFlow, ScanOPT, ScanALL and ScanNONE strategies.}
	\label{averageLCentralized}
\end{figure*}

\begin{figure*}[!h]
	\centering     
	\includegraphics[scale=0.135]{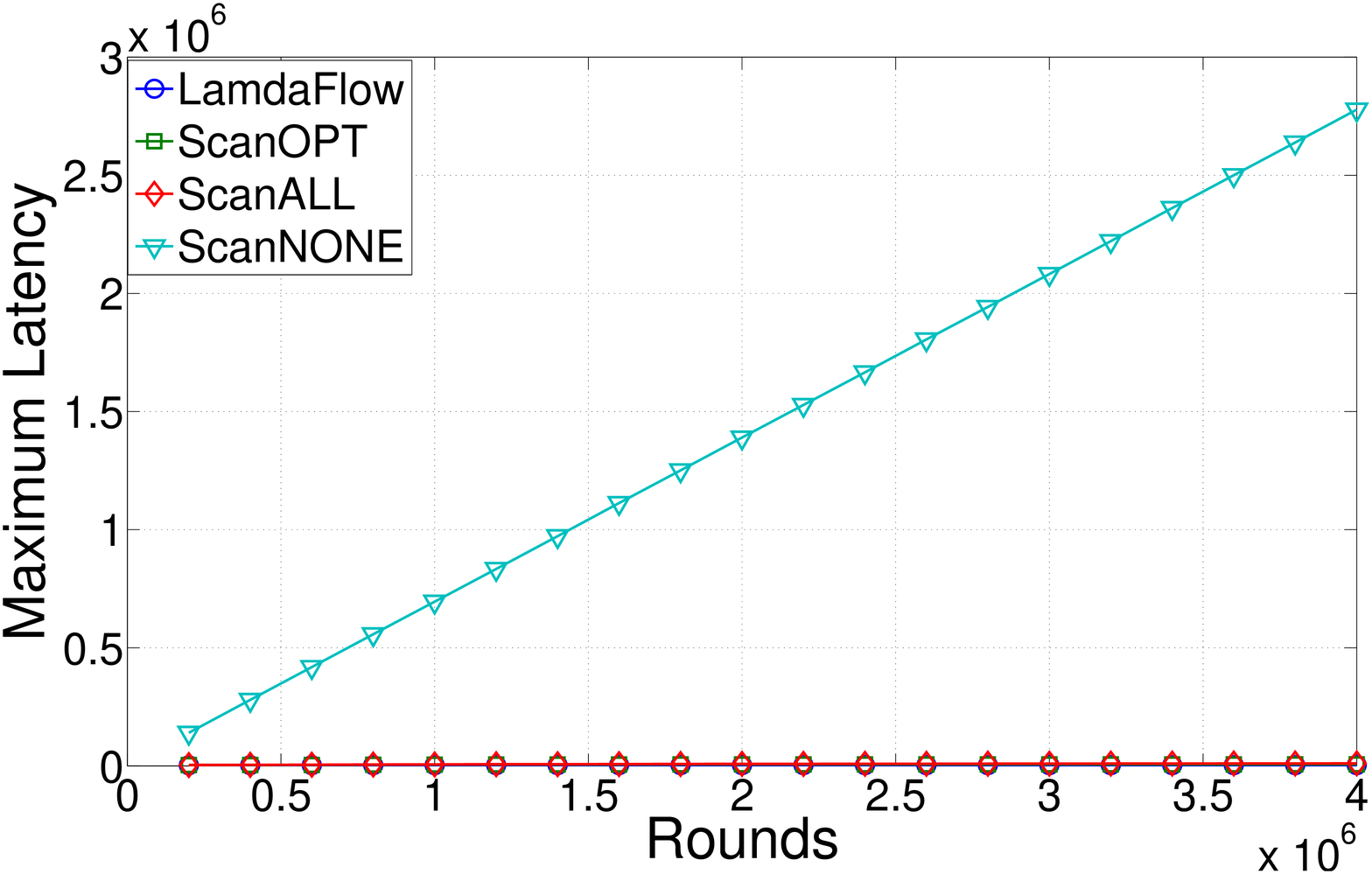}
	\includegraphics[scale=0.135]{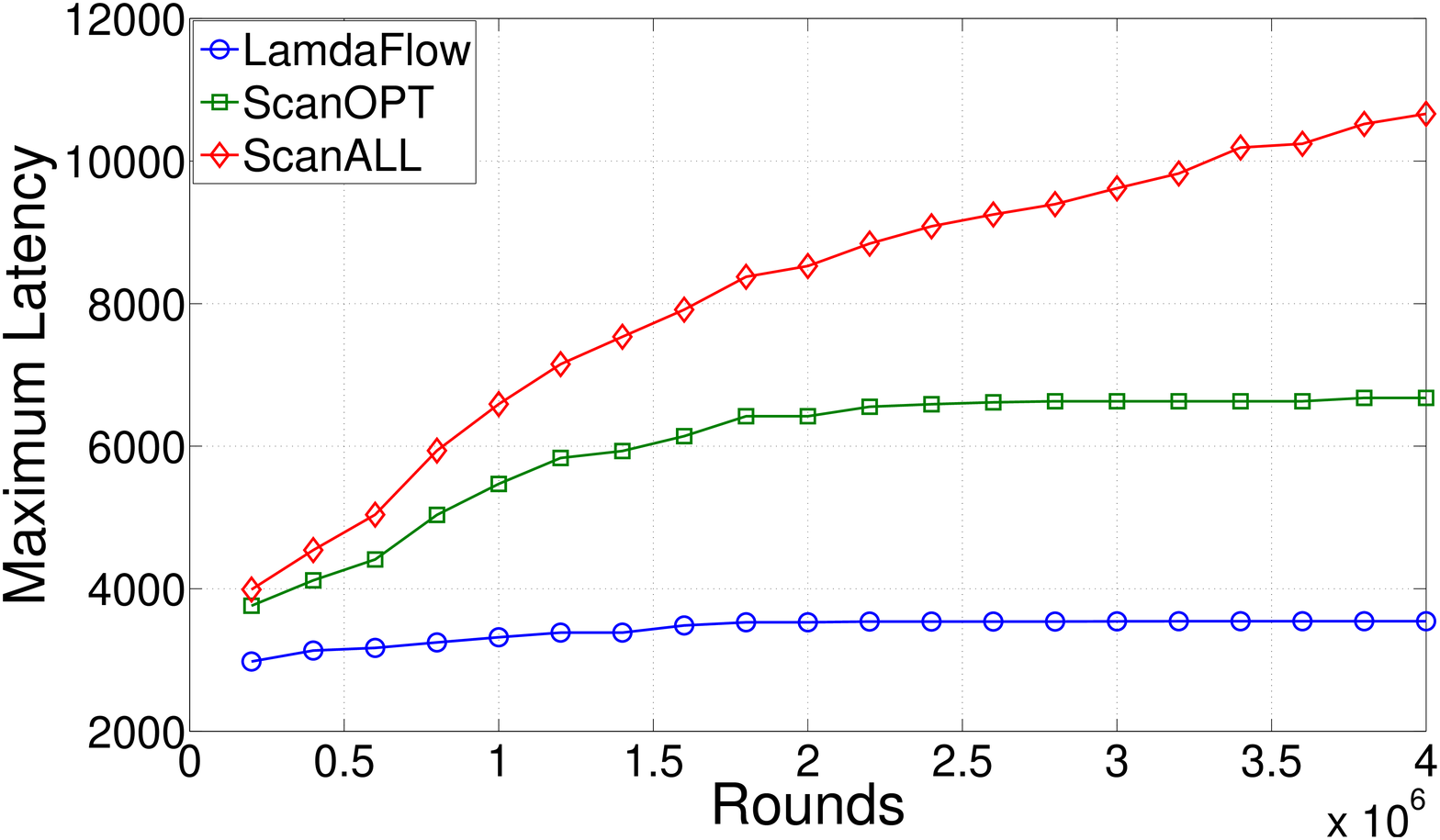}
	\caption{Comparison of maximum latency of LambdaFlow, ScanOPT, ScanALL and ScanNONE strategies.}
	\label{maximumLCentralized}
\end{figure*}

\begin{figure*}[!h]
	\centering    
	\includegraphics[scale=0.135]{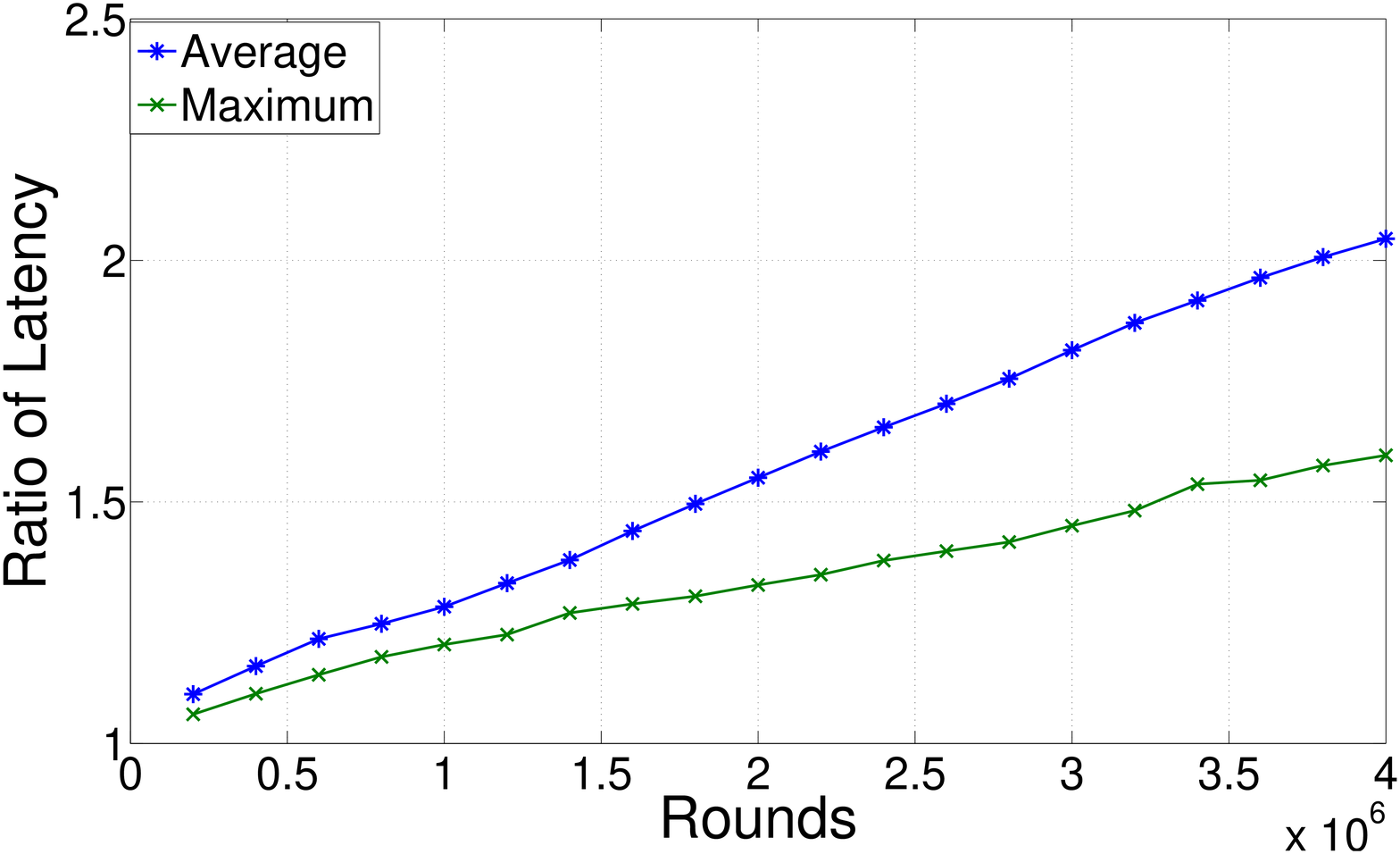}
	\includegraphics[scale=0.135]{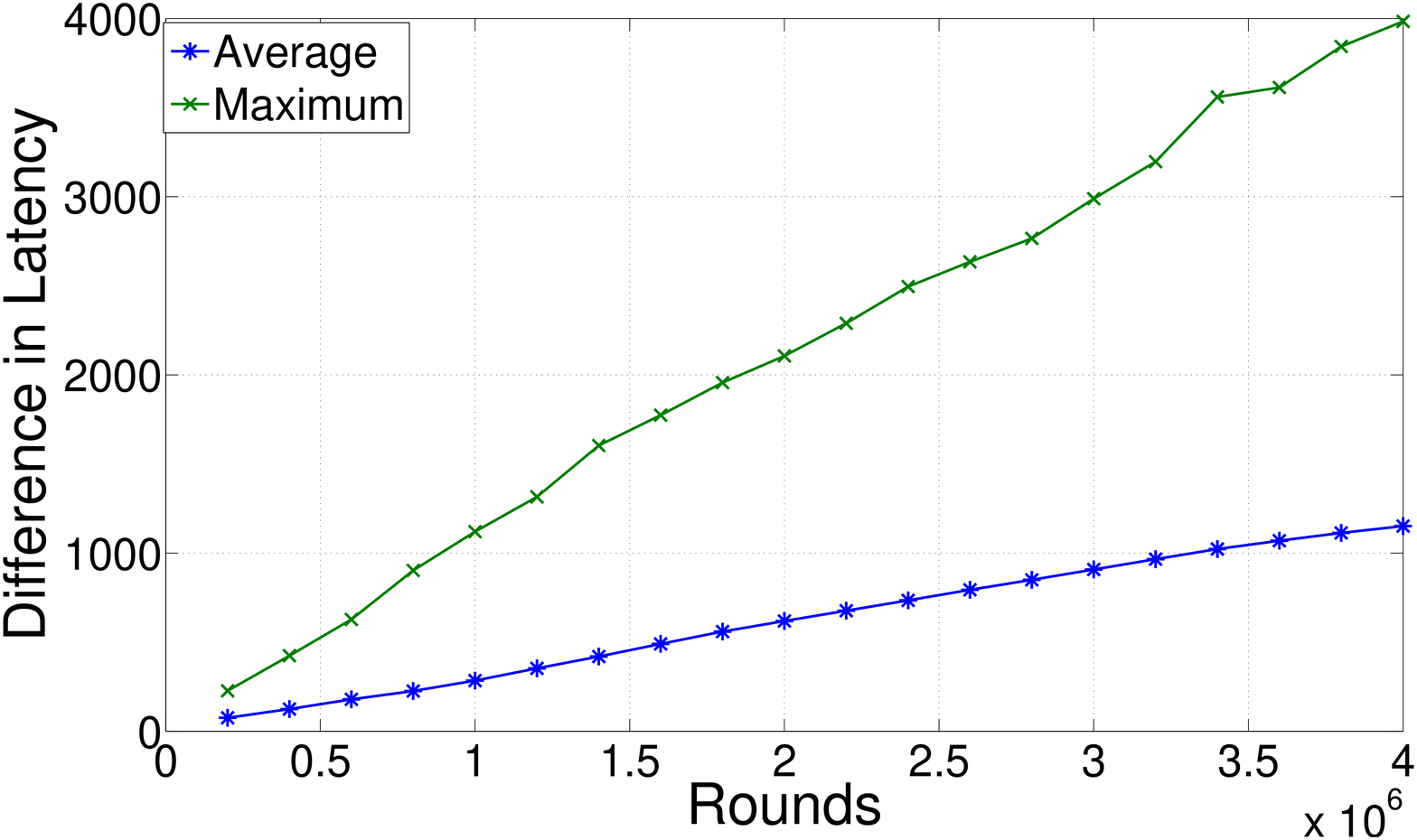}
	\caption{Ratio of ScanAll to ScanOpt latency and difference between ScanAll and ScanOpt latency (indicating ScanALL becomes worse over time).}
	\label{ratioAndDiffL}
\end{figure*}

\subsubsection{Centralized Approaches}

In order to study throughput-optimality of the scanning strategies, we recorded the latency over time for the different scanning strategies used, i.e., ScanALL, ScanOPT and ScanNONE, comparing them with the execution LambdaFlow of the genuine workload only. 
In Figure~\ref{averageLCentralized}, the average latency of the ScanNONE and ScanALL strategies grow rapidly, while it stabilizes for the ScanOPT strategy. The right part of the Figure is the zoomed left part, in order to see clearly the performance of ScanOPT versus ScanALL.
The performance of ScanALL strategy is even worse for maximum latency, where we observed that it increases rapidly; this is shown in Figure~\ref{maximumLCentralized}. This indicates that some jobs will eventually get stuck. In Figure~\ref{ratioAndDiffL} we analyzed the ratio and the difference between ScanALL to ScanOPT latencies, and both are increasing. This indicates that ScanALL is becoming worse over time,
confirming the theory that ScanOPT stabilizes while ScanALL does not
(c.f., Theorem~\ref{stability} and Theorem~\ref{non-stability}, respectively,
applied to the experiment setting of arrival rates $\lambda,\kappa$ and 
scanning vectors $\alpha^*$ and scan-all, respectively).

\begin{figure*}[!h]
	\centering 
	\includegraphics[scale=0.135]{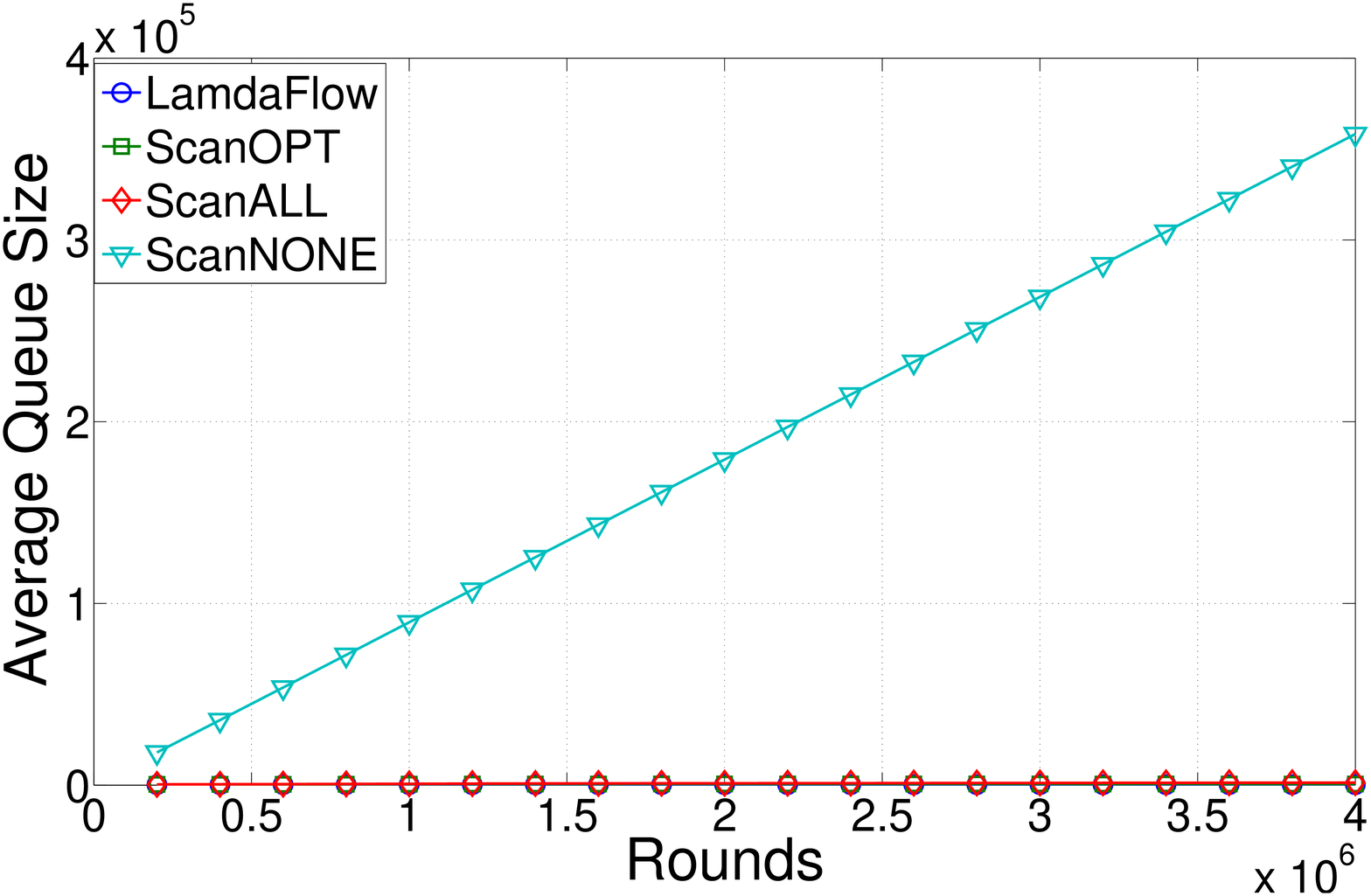}
	\includegraphics[scale=0.135]{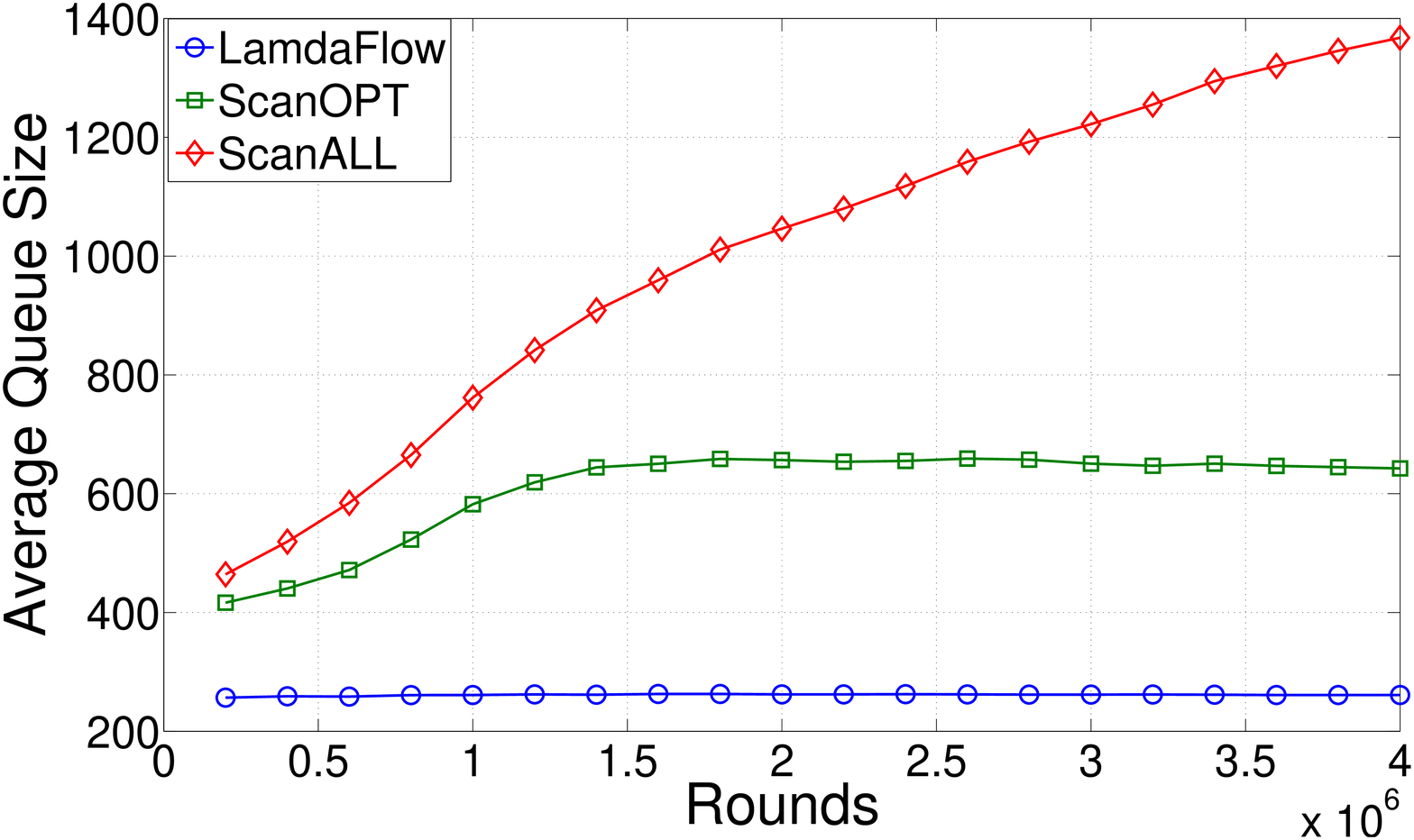}
	\caption{Comparison of average queue sizes of LambdaFlow, ScanOPT, ScanALL and ScanNONE strategies.}
	\label{averageQCentralized}
\end{figure*}

\begin{figure*}[!h]
	\centering 
	\includegraphics[scale=0.135]{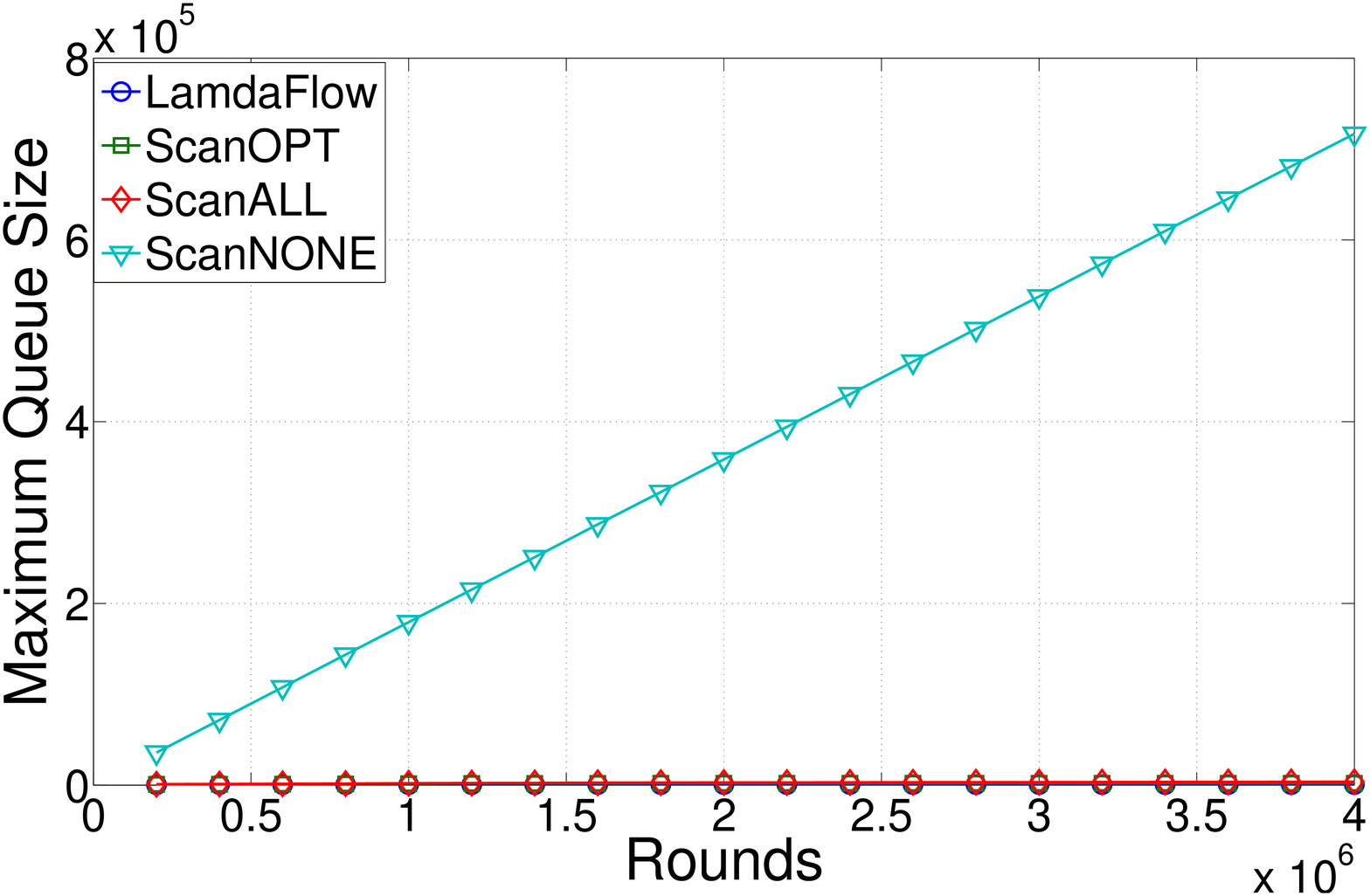}
	\includegraphics[scale=0.135]{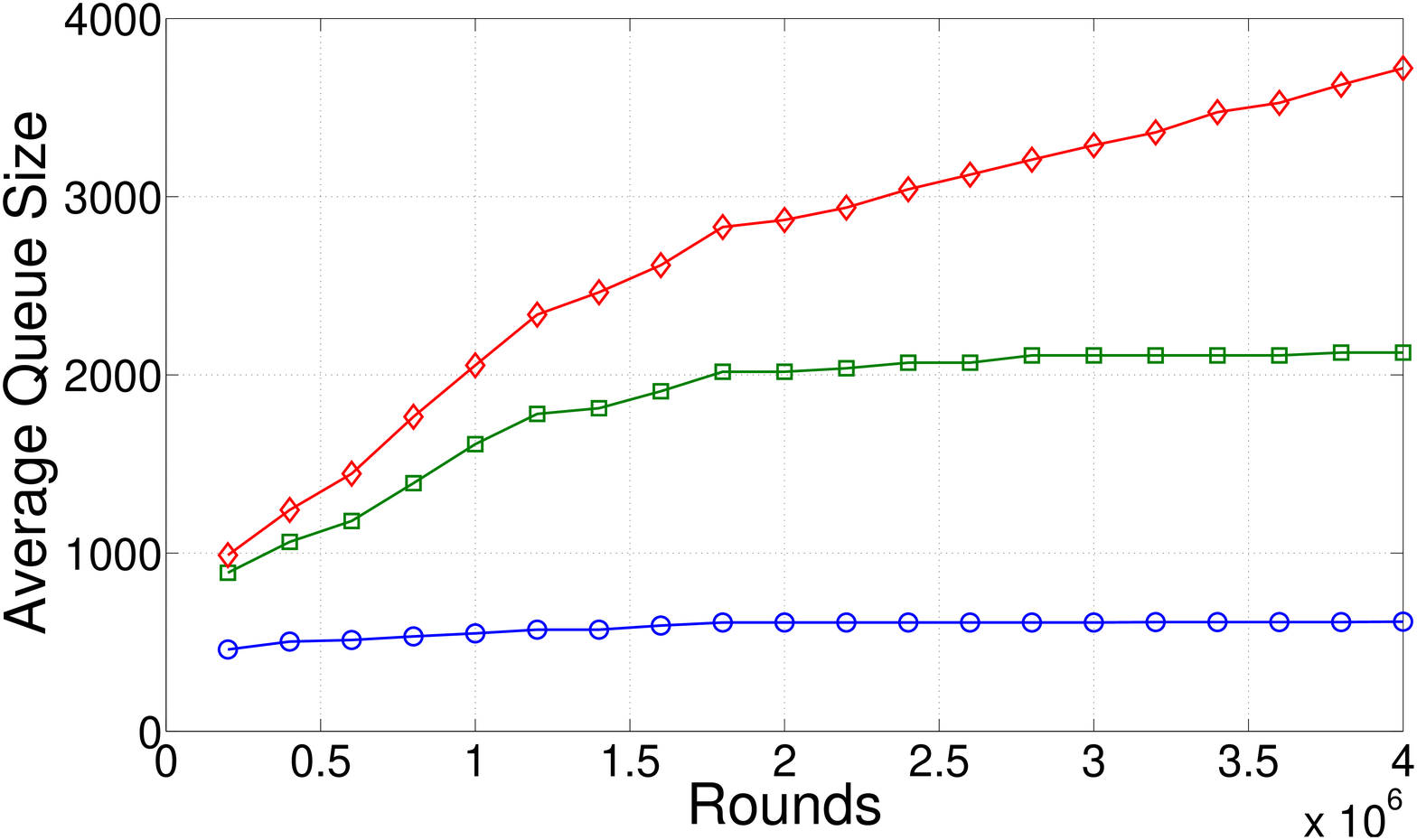}
	\caption{Comparison of maximum queue sizes of LambdaFlow, ScanOPT, ScanALL and ScanNONE strategies.}
	\label{maximumQCentralized}
\end{figure*}

\begin{figure*}[!h]
	\centering 
	\includegraphics[scale=0.135]{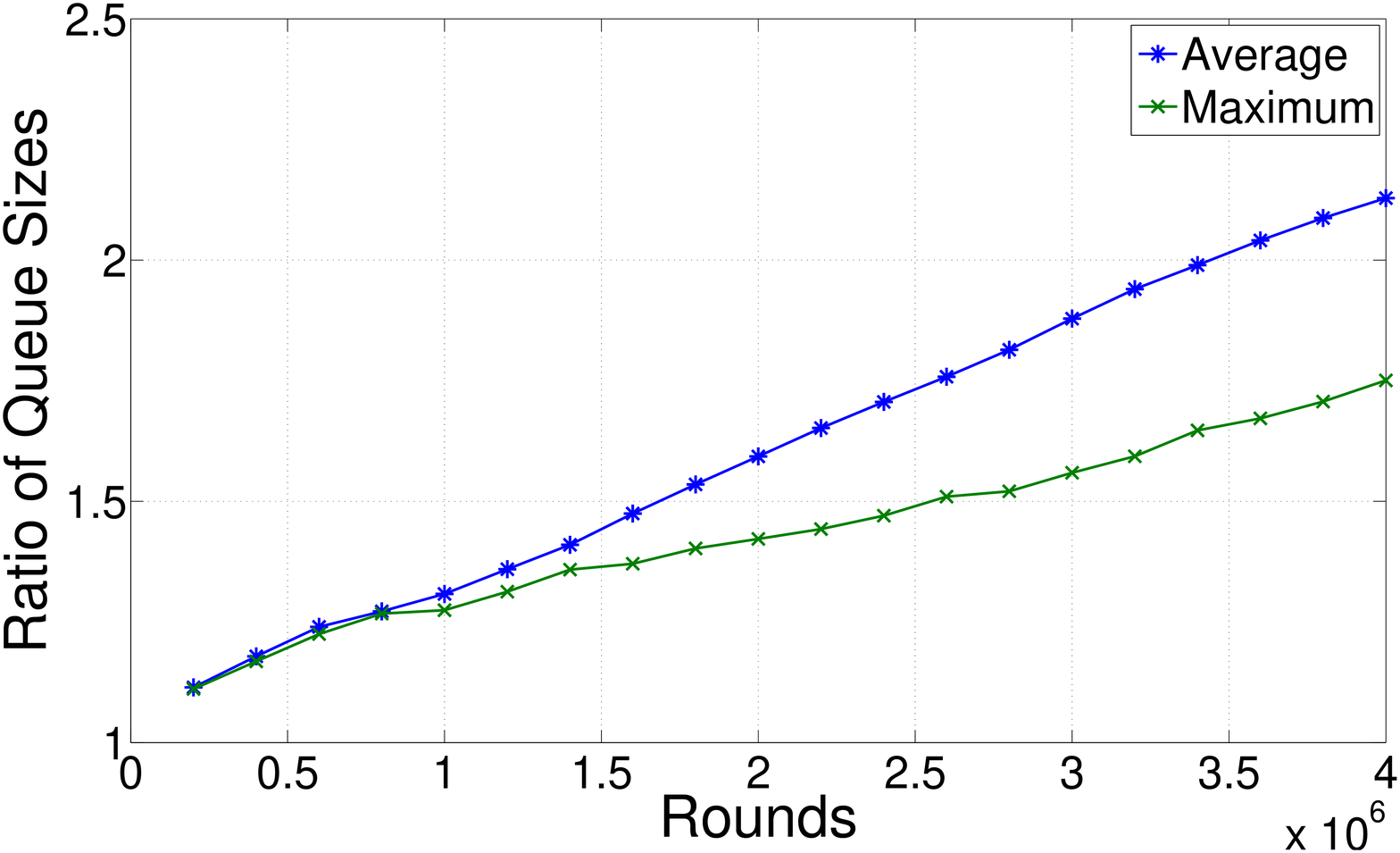}
	\includegraphics[scale=0.135]{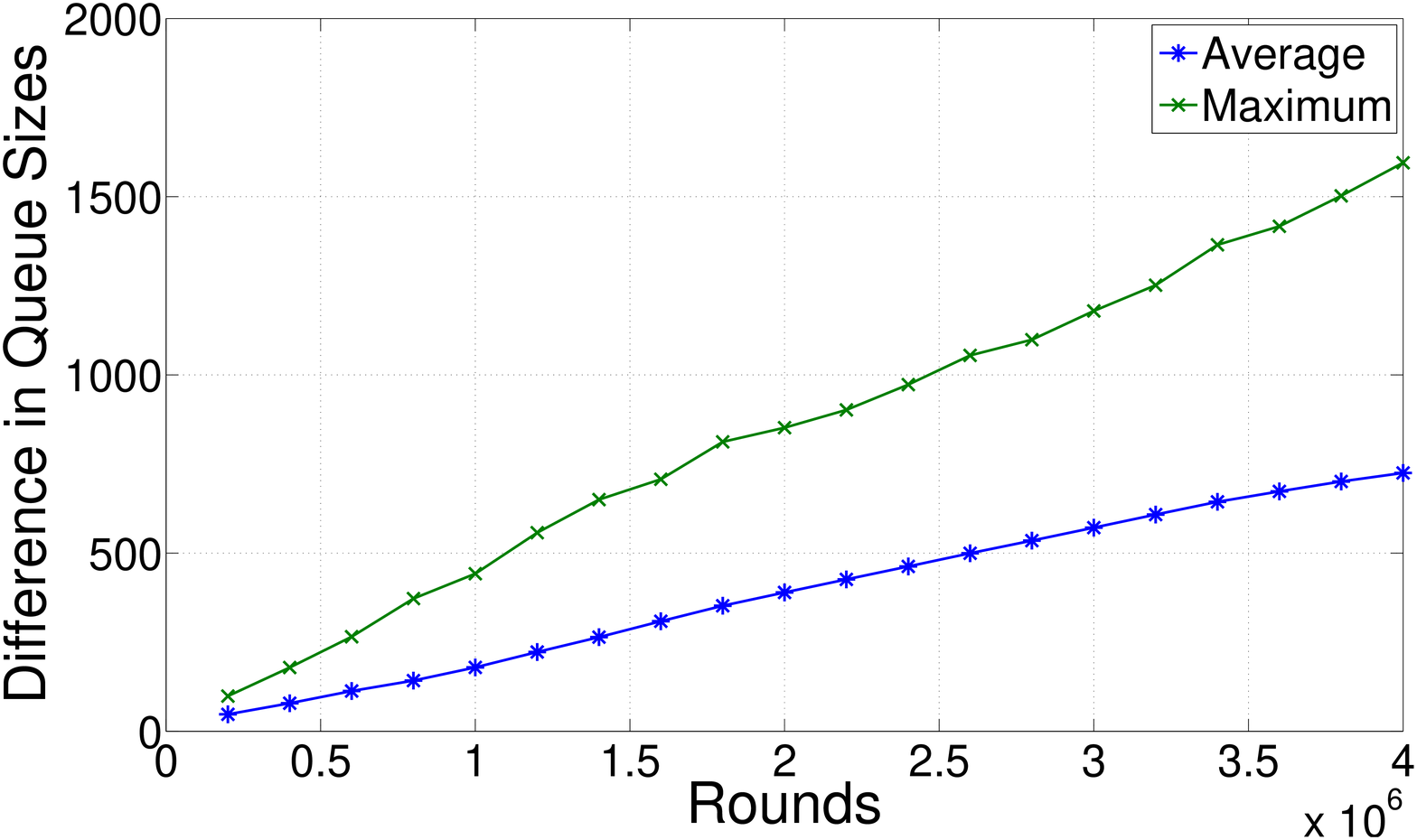}
	\caption{Ratio of ScanAll to ScanOpt Queue sizes and difference between ScanAll and ScanOpt Queue Sizes.}
	\label{ratioAndDiffQ}
\end{figure*}

\par The figures measuring queue sizes over time show that the trend is in fact similar to the trend in latency with ScanALL performing considerably worse than ScanOPT while ScanNone grow even more rapidly over time. Figure ~\ref {averageQCentralized} shows the average while queue sizes while Figure ~\ref {maximumQCentralized} shows the maximum queue sizes and Figure ~\ref{ratioAndDiffQ} shows the ratio and differences of ScanALL and ScanNONE increasing overtime.

\begin{figure*}[!h]
	\centering     
	\includegraphics[scale=0.135]{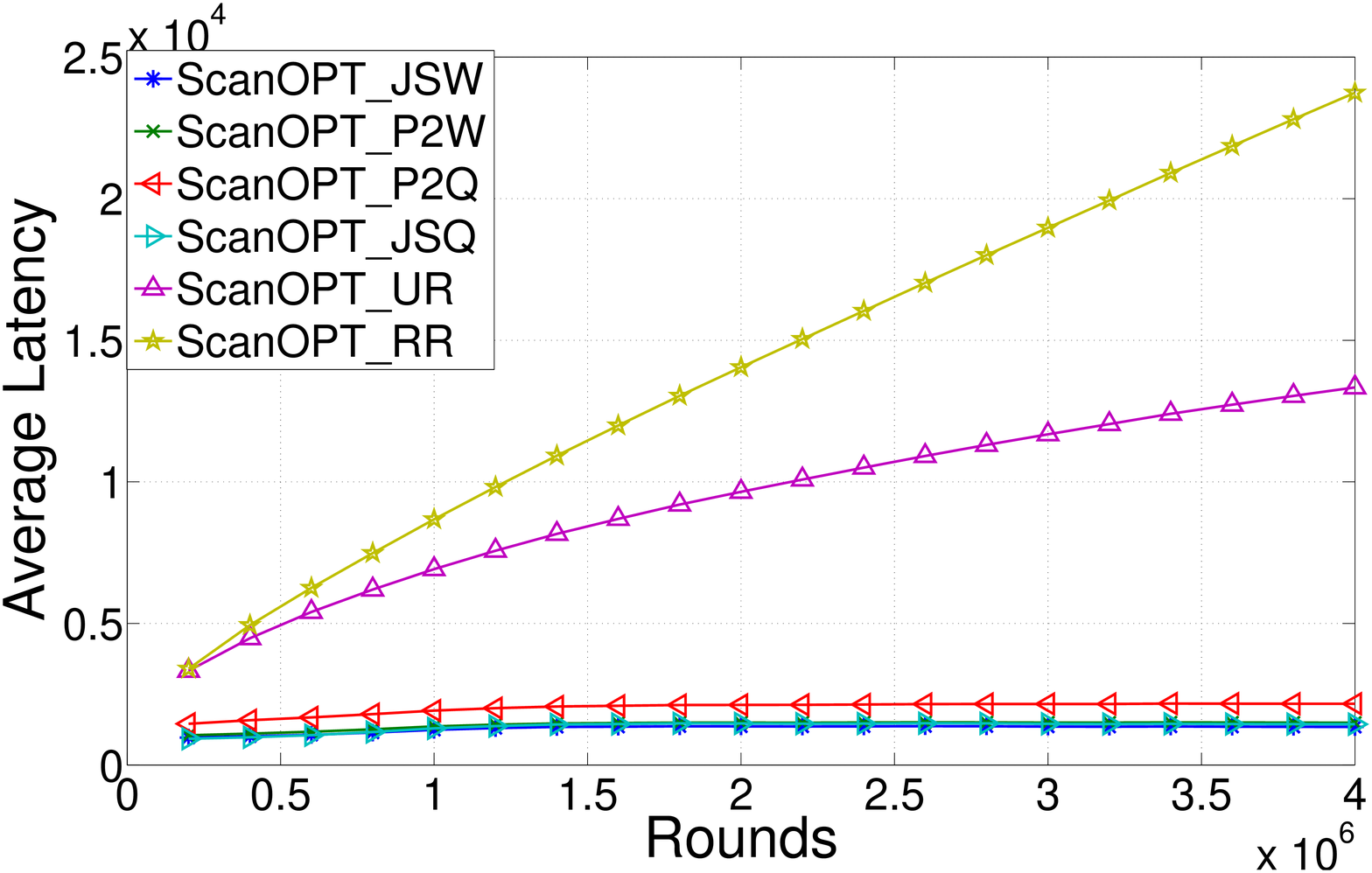}
	\includegraphics[scale=0.135]{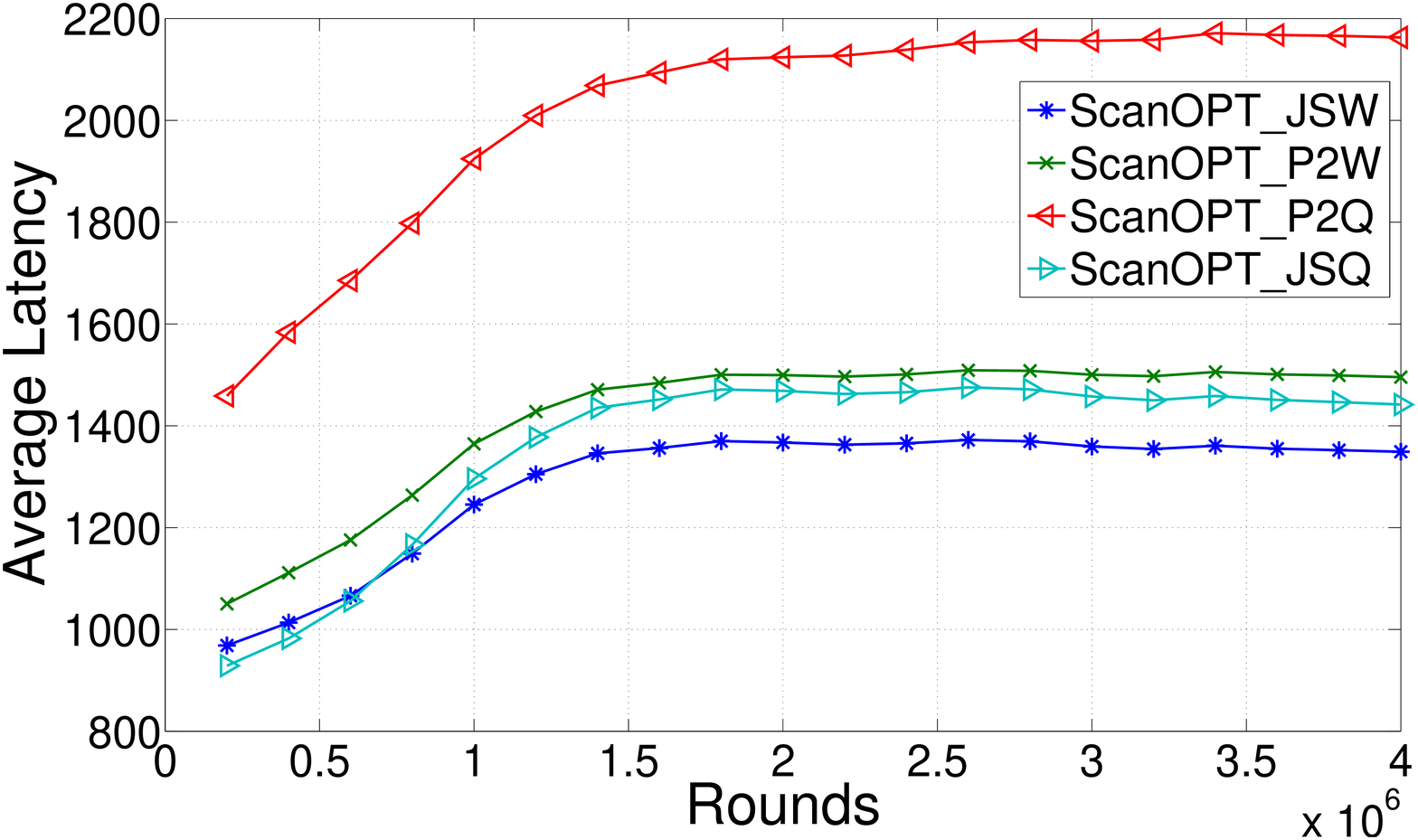}
	\caption{Comparison of average latency using ScanOPT\_JSQ, ScanOPT\_JSW, ScanOPT\_P2Q, ScanOPT\_P2W, ScanOPT\_UR and ScanOPT\_RR.}
	\label {averageLDistributed}
\end{figure*}

\begin{figure*}[!h]
	\centering     
	\includegraphics[scale=0.135]{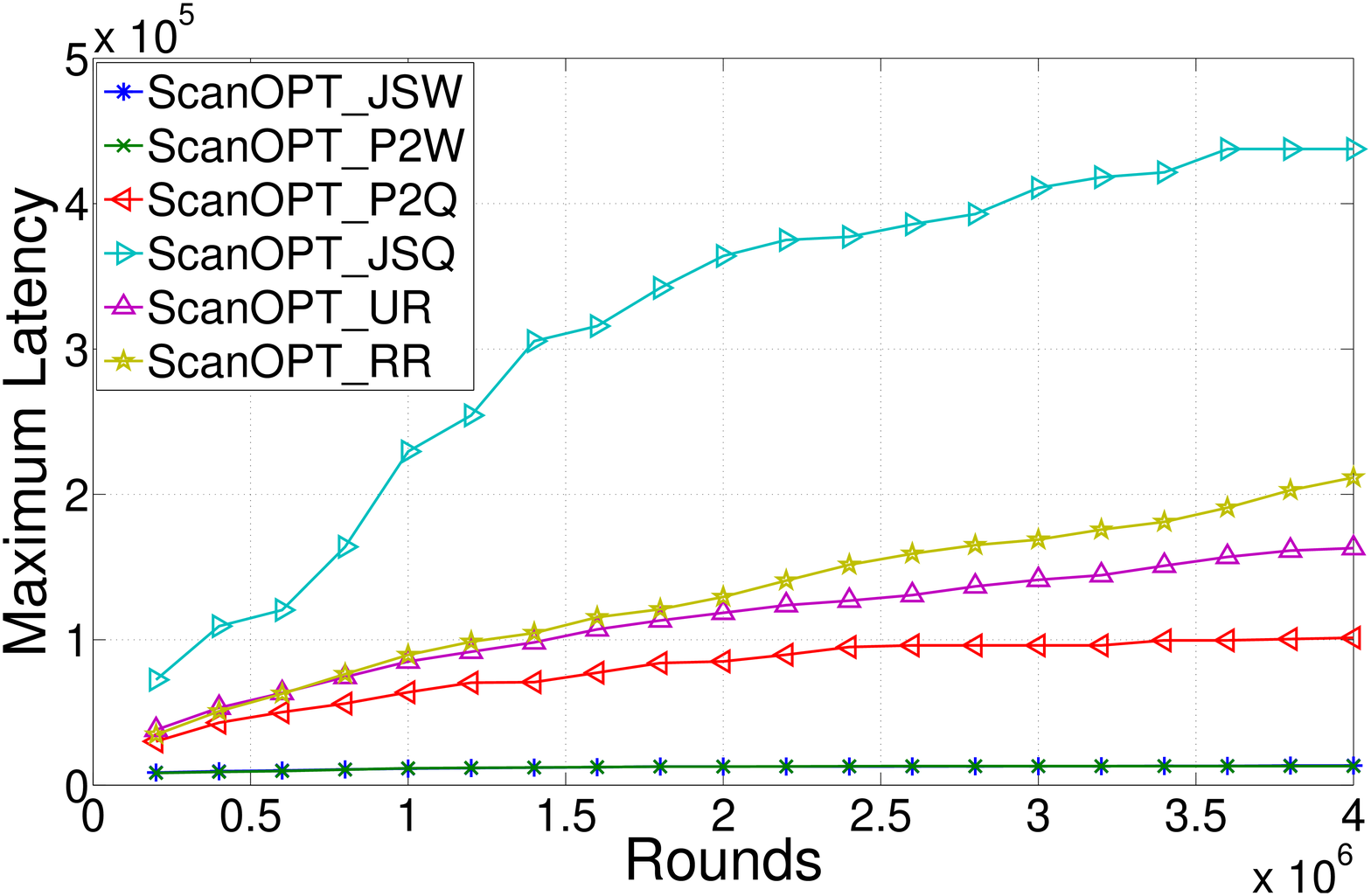}
	\includegraphics[scale=0.135]{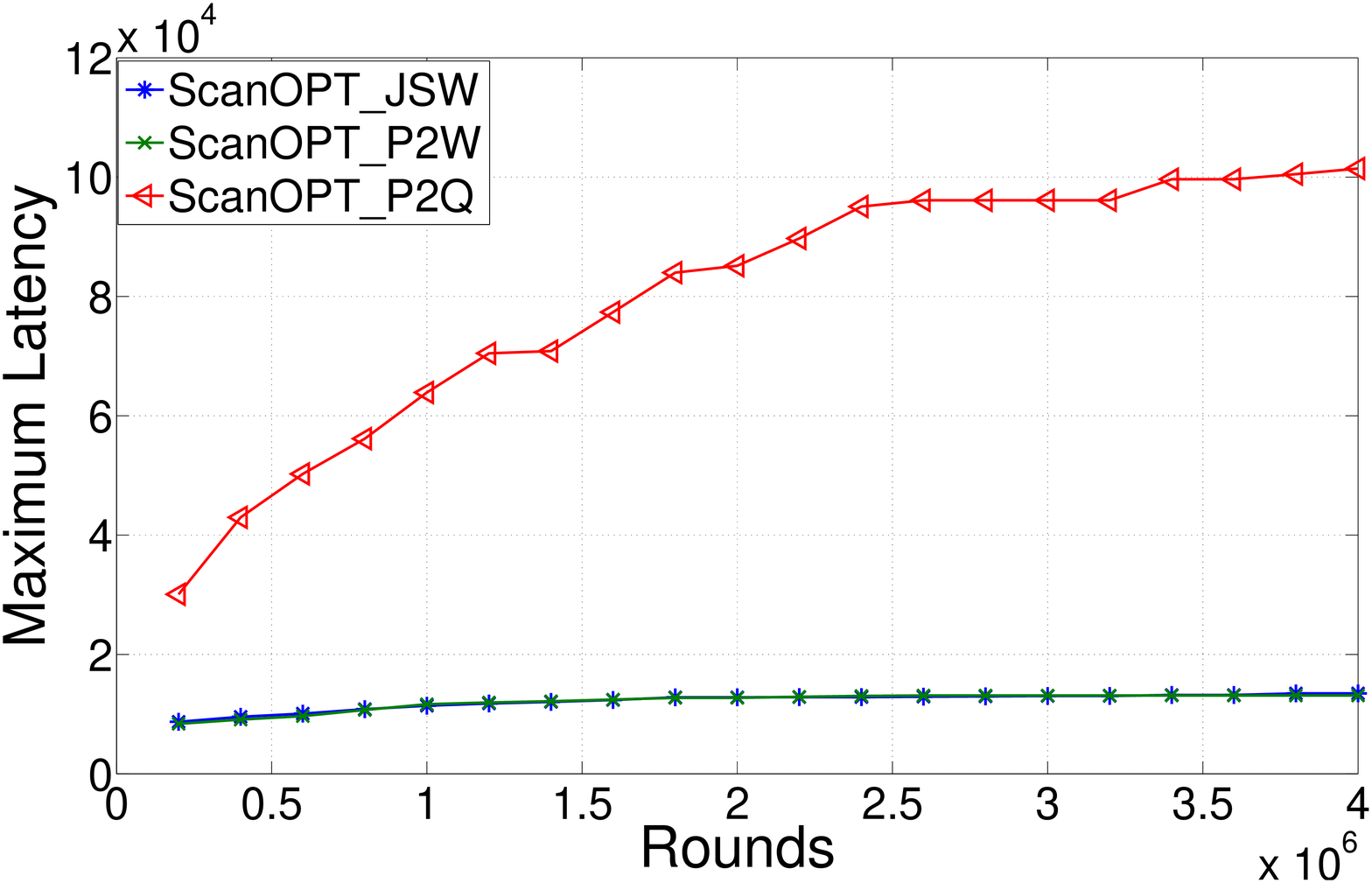}
	\caption{Comparison of maximum latency using ScanOPT\_JSQ, ScanOPT\_JSW, ScanOPT\_P2Q, ScanOPT\_P2W, ScanOPT\_UR and ScanOPT\_RR.}
	\label {maximumLDistributed}
\end{figure*}

\begin{figure*}[!h]
	\includegraphics[scale=0.135]{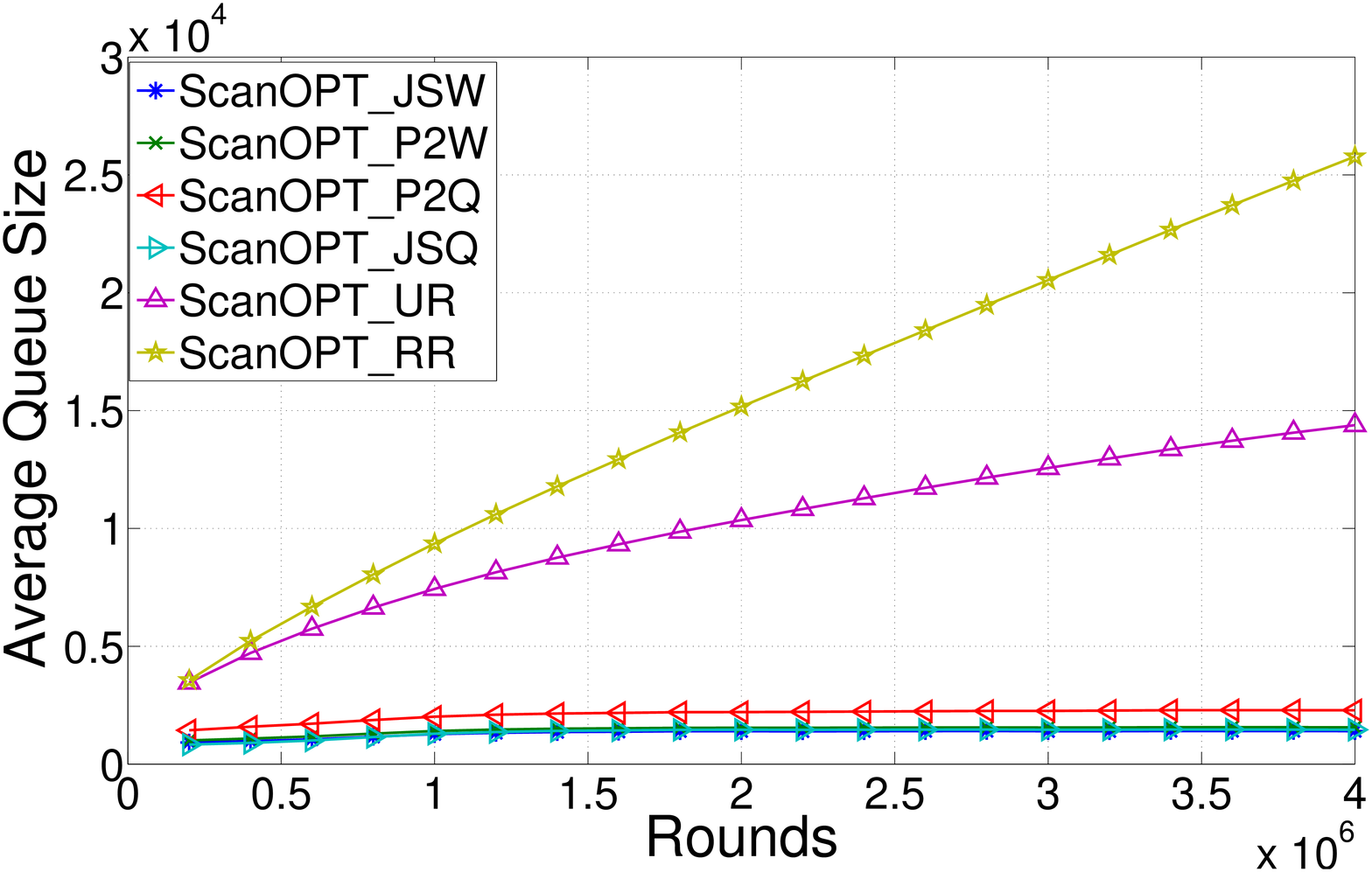}
	\includegraphics[scale=0.135]{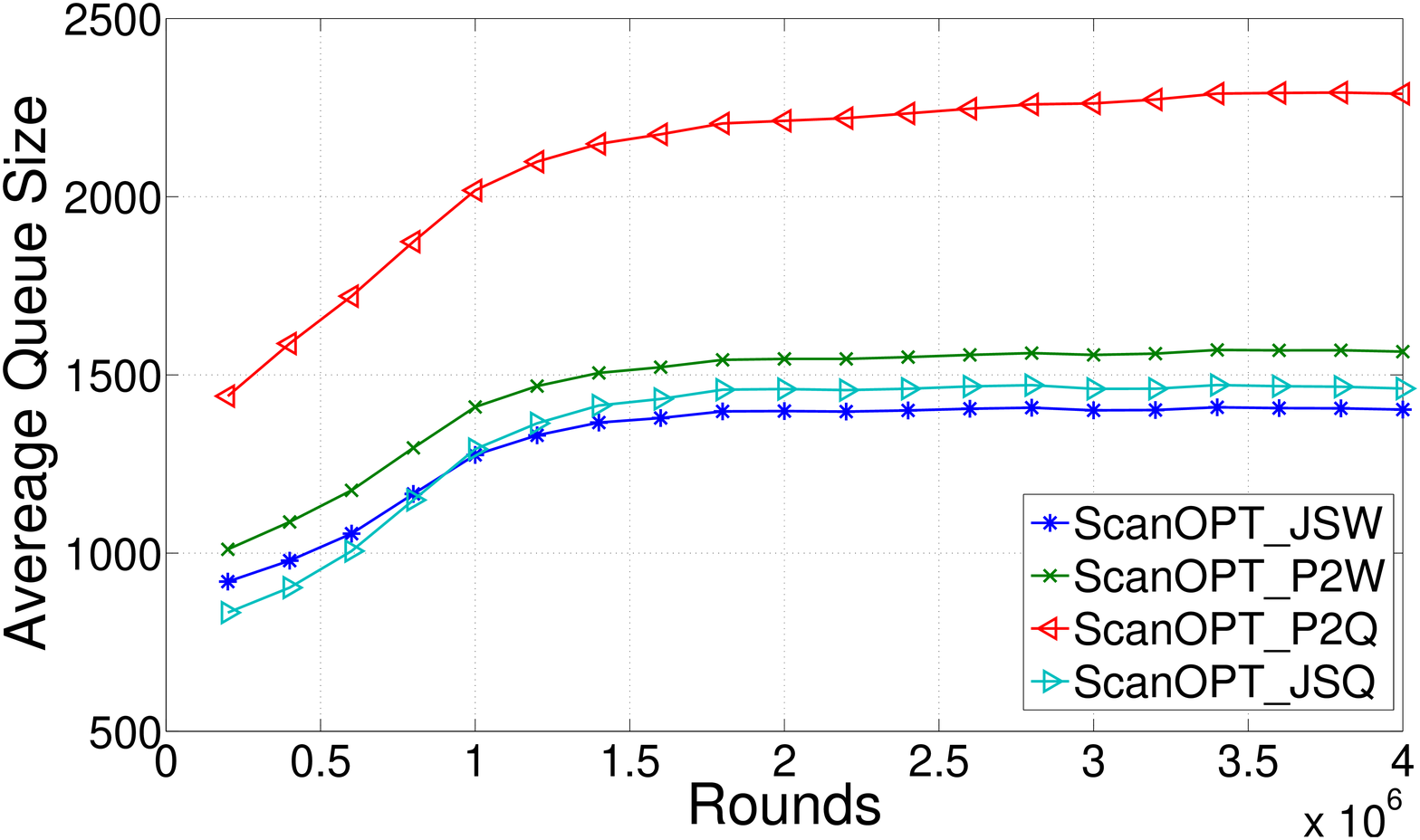}
	\caption{Decentralized Algorithms: Comparison of average queue sizes using ScanOPT\_JSQ, ScanOPT\_JSW, ScanOPT\_P2Q, ScanOPT\_P2W, ScanOPT\_UR and ScanOPT\_RR.}
	\label {averageQDistributed} 
\end{figure*}

\begin{figure*}[!h]
	\includegraphics[scale=0.135]{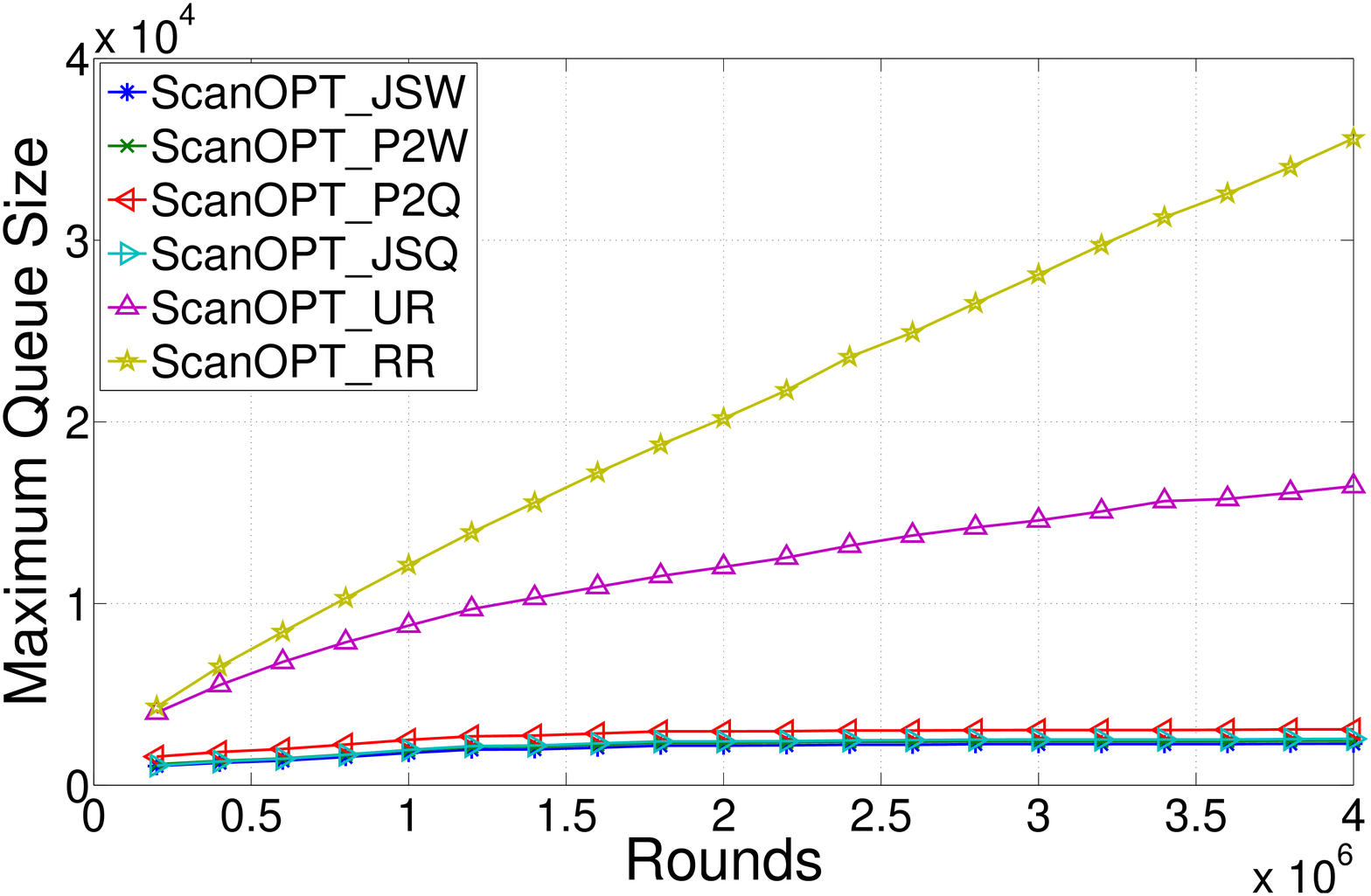}
	\includegraphics[scale=0.135]{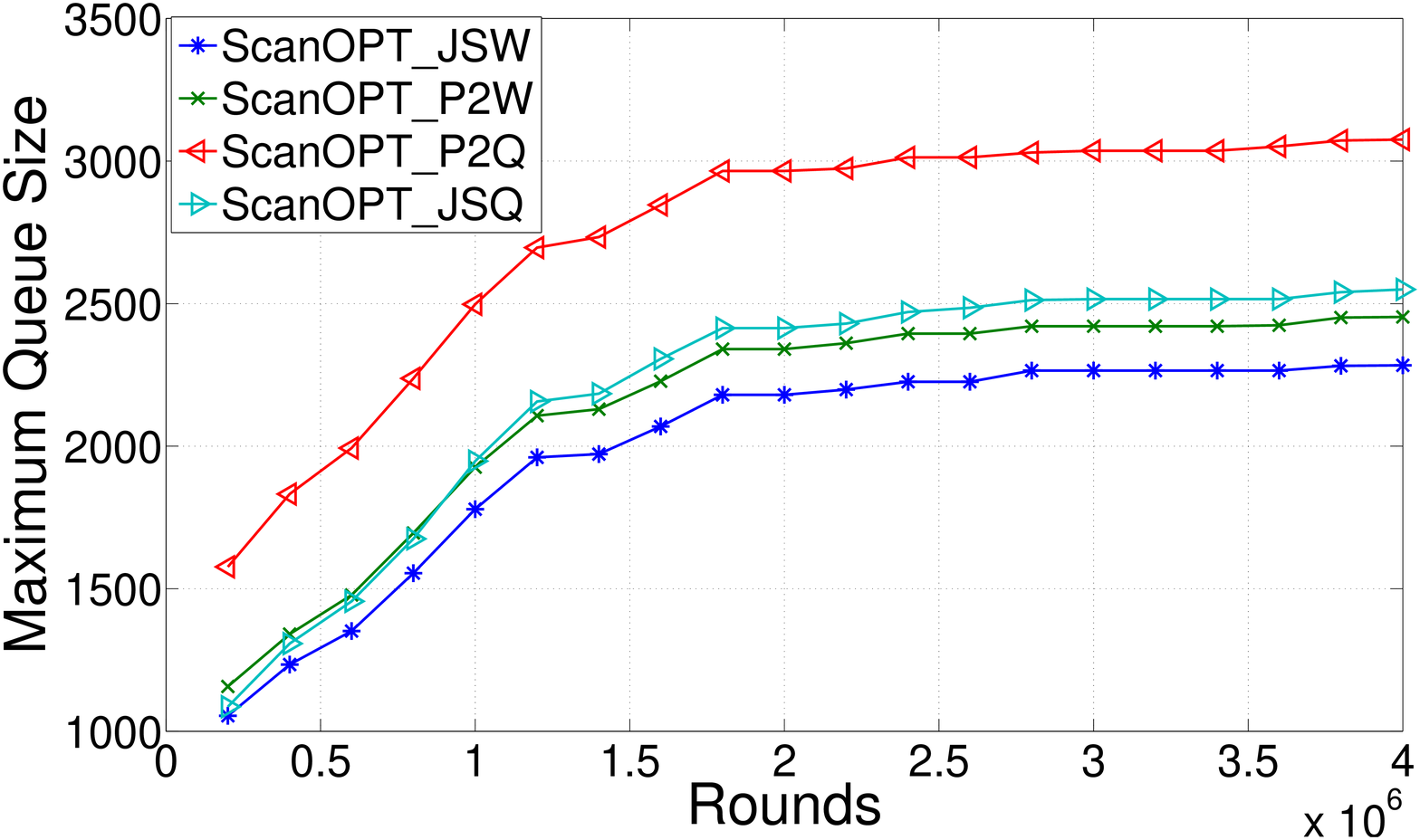}
	\caption{Decentralized Algorithms: Comparison of maximum queue sizes using ScanOPT\_JSQ, ScanOPT\_JSW, ScanOPT\_P2Q, ScanOPT\_P2W, ScanOPT\_UR and ScanOPT\_RR.}
	\label {maximumQDistributed}
\end{figure*}

\subsubsection{Decentralized Approaches}
\par In Figure~\ref{averageLDistributed}, we compare the latency of the six decentralized algorithms using ScanOPT strategy, i.e., ScanOPT\_JSQ, ScanOPT\_JSW, ScanOPT\_P2Q, ScanOPT\_P2W, ScanOPT\_UR and ScanOPT\_RR.
The best performing algorithm is the one based on JSW. 
This is followed by the two algorithms based on power of choices, and then the one based on JSW.
The worst performing algorithms are the ones based on round robin and uniform random selection, which grow rapidly.
\par

\par Figure~\ref{maximumLDistributed}, shows the trend for maximum latency. Where as expected, the algorithm based on JSW outperforms all the algorithms. A strange phenomenon we noticed is that of JSQ. We do not have a clear explanation of this phenomenon, although we suspect that
this could be because choosing right configuration based on workload, 
as is done by {\SecureMaxWork}, causes long windows of time without feeding the local queues using the shortest workload policy could make these
windows (and thus latencies) even longer than using the shortest work paradigm.

\par
In Figure \ref {averageQDistributed} we compared the average queue sizes of the 6 decentralized algorithms using ScanOPT strategy i.e. ScanOPT\_JSQ, ScanOPT\_JSW, ScanOPT\_P2Q, ScanOPT\_P2W, ScanOPT\_UR and ScanOPT\_RR. The best performing algorithm is the one based on JSW it is not surprising because it considers the exact amount of work left in all the servers to decide where to route new jobs. This is followed by JSQ and then the two algorithms based on power of two choices. The worst performing algorithms are the ones based on round robin and uniform random selection which are not stabilizing. The performance is very similar for maximum queue size comparison shown in Figure \ref {maximumQDistributed}.

\par Therefore the general trend is that the algorithms based on checking all the servers (JSQ and JSW) always outperform the ones based on power of two choices (P2W and P2Q), this phenomenon is not surprising because the algorithms based on all the servers consider the entire system state while the ones based on power of two choices are random. It should be note that the algorithms based on power of two choices are faster because decision can be made in constant time (which server to send an arriving job) while the ones based on all servers decision can only be made linear to the number of servers.


\section{Conclusions, Extensions and Open Problems}
\label{sec:miscellaneous}

\subsection{Decentralized scheduling}

We provided a rigorous mathematical analysis of a centralized scheduler {\SecureMaxWork} with central queues, where jobs could be distributed to various machines at any time after their arrival. Similar analysis applies to the decentralized {\SecureMaxWork} with Join-Shortest-Work (JSW) routing policy. Recall that this policy, upon arrival of type-$j$ job, sends it to machine that has minimum workload of type-$j$ jobs, $Z_j^{(m)}$, which is defined as in section \ref{sec:model} but now computed for each machine separately. Then each machine tries to maximize work done, $\max_{N \in S} \sum_{j=0}^J Z_j^{(m)}(t) N_j$. All the steps in the analysis of the main algorithm {\SecureMaxWork} apply in this setting, resulting in almost identical analysis as in Section~\ref{sec:analysis}.
An interesting open problem is to analyze mathematically the other five decentralized
implementations of {\SecureMaxWork}, and perhaps other similar decentralized algorithms.

\subsection{Non-preemptiveness}

In the centralized {\SecureMaxWork} we assumed that in regular time intervals (at the beginning of each time step) all machines can be reconfigured --- all jobs could be rescheduled and redistributed among the machines, where they will be further processed. In practice, interrupting execution of some jobs may be very costly. One may consider a model, where a job, once started on a machine, can not be paused or rescheduled for completion in a different time, nor processed on a different machine. 
The main idea in adapting {\SecureMaxWork} to this model is to divide time into windows of length of $T$ time slots. 
$T$ should be large enough so that any job could be started at the start of the time window without breaking the above constraint.
The algorithm will schedule jobs as previously with an additional constraint that only jobs that can be finished within a time window may be started. The stability analysis
should remain the same, except that the margin for unstable arrival rates should be made
a bit larger to accommodate potential losses of resources at the end of the time window. Thus the higher $T$ the better stability, though the latency may increase - studying
this trade-off is an interesting open problem.

\subsection{Algorithms without knowledge of arrival rates}

Finding the optimal scanning frequencies requires the knowledge of arrival rates of jobs of each type, length and genuine/malicious status. In practice, however, these values are not provided in advance. We can estimate them given a large enough sample.

If arrival rates are not provided, we can start {\SecureMaxWork} algorithm using scan-all strategy for a fixed but sufficiently long amount of time. During this time we  learn the genuine/malicious status of jobs (due to the scan-all strategy) and therefore we will be able to estimate user-generated and malicious jobs arrival rates. 
We can then use scanning frequencies that are optimal for the estimated arrival rates. Note that using a different scanning strategy at the beginning for a fixed amount of time should not have impact on stability.

Furthermore, we can use scan-all strategy for a fixed amount of time repeatedly, with significantly longer pause after each time (during which we will be using the scanning probabilities computed based on the estimates), in order to enhance the quality of  estimation of arrival rates, and thus using resources more and more efficiently. If pauses get long enough, this strategy should give better results than running scan-all strategy only once.

Another approach is to use scan-all strategy once and then run the algorithm with optimal scanning frequencies for calculated estimations, but utilize information given by scanning jobs according to optimal scanning frequencies. Since job arrivals are i.i.d. among time slots, scanning $x$ jobs randomly should give as good estimations as scanning first $x$ jobs. Therefore, even when using optimal scanning frequencies, we can improve our estimation of arrival rates each time a job was scanned.
Designing and analyzing a stable algorithm for more scarce adversarial arrivals of malicious jobs is an interesting open problem.

\subsection{Probability of successful scanning}

We can also enhance the model by assuming that scanning can fail with some probability $p$, i.e., a malicious job may be scanned but still not discovered as a malicious one. In such model, if scanning failed on a malicious job, then this job is indistinguishable from a genuine user-generated one. Therefore with probability $p$ scanned malicious jobs will have to be processed fully. We should still be able to use {\SecureMaxWork} algorithm, with a slight modification of the formula for job weights $Z_j$, which would have to include the failure probability as additional factor in some components. As the result, the optimal scanning frequencies may be slightly different but still in many cases better than scan-all or scan-none.

\bibliographystyle{plain}

\end{document}